\documentclass[a4paper,12pt]{article}

\usepackage[lmargin=2cm,rmargin=2cm]{geometry}
\usepackage[utf8]{inputenc}
\usepackage{amsmath,amssymb,amsfonts,mathtools,amsthm,stmaryrd,bbm,physics, dsfont}
\usepackage[group-separator={,}]{siunitx}
\usepackage{todonotes}
\usepackage{booktabs}
\usepackage{afterpage}
\usepackage{url}
\usepackage{hyperref}

\usepackage{comment}
\usepackage{bm}
\usepackage{subcaption}
\usepackage{graphicx}

\usepackage{algorithm}       
\usepackage{algpseudocode}   

\theoremstyle{definition}

\newtheorem{lemma}{Lemma}
\newtheorem{proposition}{Proposition}
\newtheorem{definition}{Definition}

\newtheorem{problem}{Problem}
\newtheorem{observation}{Observation}

\usepackage{xspace}
\usepackage{hyperref}
\usepackage[capitalize]{cleveref}
\usepackage{authblk}

\usepackage{xcolor,color}

\newcommand{\f}{f}
\newcommand{\fdag}{f^{\dagger}}

\DeclareMathOperator{\SO}{SO}
\DeclareMathOperator{\pf}{pf}
\DeclareMathOperator{\sgn}{sgn}
\DeclareMathOperator{\U}{U}

\newcommand{\mmd}{$\rm{MMD}^2$}
\newcommand{\x}{\bm{x}}
\newcommand{\y}{\bm{y}}
\newcommand{\q}{q_{\mathbf{w}}}
\newcommand{\p}{p}
\newcommand{\fbm}{FBM}
\newcommand{\uflo}{U_{\rm{FLO}}}

\newcommand{\conj}[1]{\overline{#1}}

\usepackage[
    style=numeric-comp,
    sorting=none,
    defernumbers=true,
    maxbibnames=20,
    maxcitenames=20,
    date=year,
    isbn=false,
    giveninits=true,
    uniquename=init,
]{biblatex}

\renewbibmacro*{doi+eprint+url}{%
    \iftoggle{bbx:url}
    {\iffieldundef{doi}{\iffieldundef{eprint}{\usebibmacro{url+urldate}}{}}{}}
    {}%
    \newunit\newblock
    \iftoggle{bbx:eprint}
    {\iffieldundef{doi}{\usebibmacro{eprint}}{}}
    {}%
    \newunit\newblock
    \iftoggle{bbx:doi}
    {\printfield{doi}}
    {}
}

\renewbibmacro{in:}{}
\AtEveryBibitem{%
\clearfield{issue}%
\clearfield{number}}

\DefineBibliographyStrings{english}{%
    urlseen = {last accessed on},
}

\title{Fermionic Born Machines: Classical training of quantum generative models based on Fermion Sampling}

\author[1,2]{Bence Bak\'o}
\author[1,2]{Zolt\'an Kolarovszki}
\author[1,3,4]{Zolt\'an Zimbor\'as}
\affil[1]{HUN-REN Wigner Research Centre for Physics, Budapest, Hungary}
\affil[2]{E\"otv\"os Lor\'and University,  Budapest, Hungary}
\affil[3]{Algorithmiq Ltd, Kanavakatu 3C 00160 Helsinki, Finland}
\affil[4]{University of Helsinki, Yliopistonkatu 4 00100 Helsinki, Finland}

\date{}

\begin{document}
\maketitle
\begin{abstract}
Quantum generative learning is a promising application of quantum computers, but faces several trainability challenges, including the difficulty in experimental gradient estimations. For certain structured quantum generative models, however, expectation values of local observables can be efficiently computed on a classical computer, enabling fully classical training without quantum gradient evaluations. Although training is classically efficient, sampling from these circuits is still believed to be classically hard, so inference must be carried out on a quantum device, potentially yielding a computational advantage. In this work, we introduce Fermionic Born Machines as an example of such classically trainable quantum generative models. The model employs parameterized magic states and fermionic linear optical (FLO) transformations with learnable parameters. The training exploits a decomposition of the magic states into Gaussian operators, which permits efficient estimation of expectation values. Furthermore, the specific structure of the ansatz induces a loss landscape that exhibits favorable characteristics for optimization.
The FLO circuits can be implemented, via fermion-to-qubit mappings, on qubit architectures to sample from the learned distribution during inference. Numerical experiments on systems up to 160 qubits demonstrate the effectiveness of our model and training framework.
\end{abstract}

\section{Introduction}
    The success of classical deep learning stems from the surprisingly good trainability of deep neural networks, despite their complexity and nonlinearity~\cite{Goodfellow-et-al-2016}. A key factor contributing to good training properties is the availability of efficient gradient computation techniques, allowing effective optimization even in high-dimensional parameter spaces~\cite{baydin2018automatic,bengio2014deep}. Contrary to this success, quantum machine learning (QML) models face additional challenges related to trainability, such as barren plateaus~\cite{mcclean2018barren, larocca2025barren}, poor local minima~\cite{anschuetz2022quantum}, and most importantly, costly gradient computation~\cite{10.5555/3666122.3668062, gilyen2019optimizing}. These problems can be addressed using structured and restricted quantum circuit models that provide better loss~\cite{PRXQuantum.5.020328, Monbroussou2025trainability, PRXQuantum.3.030341, PRXQuantum.4.010328, PRXQuantum.4.020327}, or more efficient gradient estimation strategies~\cite{coyle2025training, bowles2025backpropagation}. 
    
    An essential part of establishing good trainability is to show---empirically or analytically---that barren plateaus are absent~\cite{PhysRevX.11.041011, PhysRevLett.129.270501, Larocca2022diagnosingbarren, Park2024hamiltonian, PRXQuantum.5.030320}. However, in many cases, the provable absence of barren plateaus implies that the expectation values of (local) observables are efficiently simulable classically \cite{cerezo2025does, bermejo2024quantum, lerch2024efficient, PhysRevResearch.6.023218}. This presents a significant challenge: since the output of most standalone supervised QML models is precisely such an expectation value, these models are rendered effectively classical. While this suggests no-go theorems for supervised QML, it also opens a new question: can this efficient classical component be exploited within generative learning scenarios without making the entire model classically simulable?
    
    This leads us to consider restricted quantum circuits that are hard to sample from classically, yet allow efficient local expectation value estimation. Such circuits are ideal quantum generative models as they can be trained classically but require a quantum computer for sampling.
    This idea was explored in Refs.~\cite{kasture2023protocols, recio2025iqpopt, recio2025train} for instantaneous quantum polynomial-time (IQP) circuits, where direct estimation of probabilities and local expectation values can be performed classically.

    In this work, we consider a different class of restricted quantum circuits, namely fermionic linear optical (FLO) or matchgate circuits~\cite{terhal2002classical, knill2001fermionic, valiant2002quantum}. Equipped with magic input states, they serve as the basis of the Fermion Sampling quantum advantage scheme~\cite{oszmaniec2022fermion}, having provable hardness guarantees. Subsequent works have also examined FLO evolutions of magic states, see Refs.~\cite{diaz2023showcasing,reardonsmith2024fermioniclinearopticalextent,leimkuhler2025exponentialquantumspeedupsquantum,sierant2025fermionicmagicresourcesquantum,alam2025fermionicdynamicstrappedionquantum,PRXQuantum.6.010319}.
    Building on the results of Ref.~\cite{oszmaniec2022fermion}, we introduce \textit{Fermionic Born Machines} (FBMs) with the following favorable training and inference properties:
    \begin{itemize}
        \item The output probability distribution of FBMs with parametrized magic input states can be classically hard to sample from under reasonable complexity theoretic assumptions;
        \item The expectation value of constant-length Pauli-Z strings can be computed in polynomial time;
        \item The trained model can be sampled efficiently on a quantum computer using a linear-depth quantum circuit.
    \end{itemize}
    We prove these properties analytically and give an explicit simulation algorithm that shows better time complexity than the naive Heisenberg evolution algorithms for this specific setup~\cite{pauliprop, miller2025simulation}. The corresponding model and training method are shown in \cref{fig:architecture}. We also present numerical investigations on important aspects of this framework and demonstrate its scalability for circuits exceeding $100$ qubits. 

    The remainder of this work is structured as follows. In \cref{sec:background}, we review the relevant concepts in quantum generative model training and the Fermion Sampling scheme. In \cref{sec:fbms} we formally define FBMs and the associated classical training scheme, stating its most important properties, while deferring the proofs to \cref{app:training_strategy}. In \cref{sec:demonstrations}, we present our numerical findings that include the analysis of convergence properties and the notable effect of overparametrization, but also demonstrate the scalability of our approach with relevant use-cases. Finally, in \cref{sec:conclusion}, we conclude our study and present ideas for potential future work

\begin{figure}
    \centering
    \includegraphics[width=0.9\linewidth]{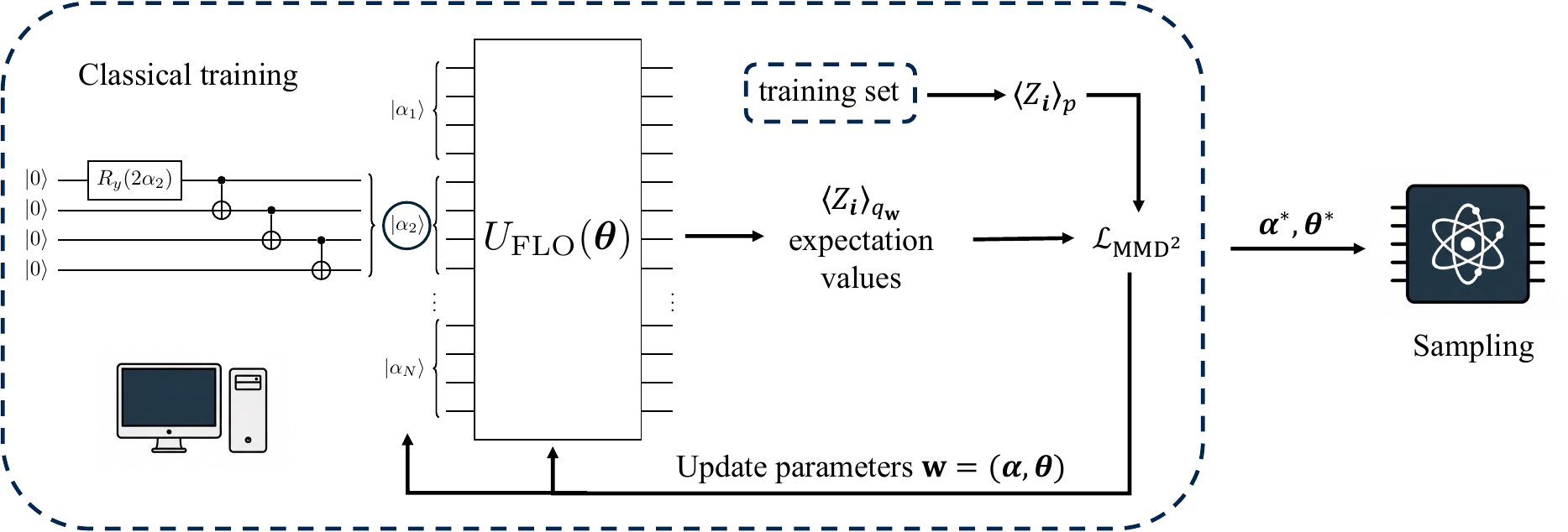}
    \caption{\textbf{Framework for the classical training of Fermionic Born Machines (FBMs).} The FBM ansatz consists of a parametrized magic input state followed by a fermionic linear optical (FLO) transformation. The parameters are classically optimized by first estimating expectation values of Pauli-Z strings for the training set and the FBM model, then computing the squared maximum mean discrepancy loss function. After training, the corresponding quantum circuit is run on a quantum hardware with the optimized parameters.
    }
    \label{fig:architecture}
\end{figure}
    
\section{Background}
\label{sec:background}
    In this section, we briefly review the basic concepts required for this work, namely, the classical training of quantum generative models and the Fermion Sampling scheme. For a more technical introduction to training quantum generative models and fermionic quantum systems, please refer to \cref{app:cl_training} and \cref{app:fermionic}, respectively.
    
    \subsection{Classical training of quantum generative models}

        The goal of generative modeling is to learn a representation of a model probability distribution $q_{\mathbf{w}}$ that is close to the target distribution $p$, i.e., samples drawn from $q_{\mathbf{w}}$ resemble the samples from $p$.
        Consequently, this representation should come with an algorithm that provides an efficient way to draw samples from $q_{\mathbf{w}}$.
        In this work, we focus on the paradigmatic framework of Quantum Circuit Born Machines (QCBMs), where efficient sampling can be achieved by a polynomial-size quantum circuit and computational basis measurements~\cite{benedetti2019generative, PhysRevA.98.062324, coyle2020born}. We consider the task of generative learning defined over $n$ binary random variables.
    
        These models are characterized by a quantum circuit ansatz $U(\mathbf{w})$, where $\mathbf{w}$ is the vector of trainable parameters. Considering the computational basis measurement on each qubit, a QCBM defines the model probability distribution as
        \begin{equation}
            q_{\mathbf{w}}(\bm{x}) = |\langle \bm{x}| U(\mathbf{w})|\bm{0}\rangle|^2,
        \end{equation}
        where $\bm{x}$ is a binary string $x_i \in \{0,1\}$. To obtain samples that resemble the training dataset, the parameters are updated to minimize a suitable loss function.

        Although explicit loss functions such as the total variation distance (TVD) or the Kullback–Leibler (KL) divergence can be used to train small-scale models, they suffer from high sample complexity and may cause trainability issues such as barren plateaus~\cite{rudolph2024trainability}. A more practical, implicit loss function is the squared \textit{maximum mean discrepancy} (\mmd), which in its canonical form takes the raw samples from $p$ and $q_{\mathbf{w}}$ as input. In Ref.~\cite{rudolph2024trainability}, it was shown that instead of relying on samples from both probability distributions, the \mmd\ loss function can be reformulated using expectation values of Pauli-Z strings as
        \begin{equation}
            \mathcal{L}_{\rm{MMD}^2}(p, q_{\mathbf{w}}) = \sum_{\bm{i} \in \mathcal{P}([n])} p_K(\bm{i}) (\expval{Z_{\bm{i}}}_p -\expval{Z_{\bm{i}}}_{q_{\mathbf{w}}})^2 = \mathds{E}_{\bm{i} \sim p_K} [\expval{Z_{\bm{i}}}_p -\expval{Z_{\bm{i}}}_{q_{\mathbf{w}}}]^2,
        \end{equation}
        where $\expval{Z_{\bm{i}}}_p$ and $\expval{Z_{\bm{i}}}_{q_{\mathbf{w}}}$ denote the expectation values of the Pauli-Z strings over the random variables for the target and model probability distribution, respectively, and $Z_{\bm{i}} \coloneqq Z_{i_1} \cdots Z_{i_{\ell}}$ with $\ell$ being the length of $\bm{i}$. In the above expression, $K$ refers to the choice of the kernel function and $p_K$ denotes the corresponding probability distribution over the possible Pauli-Z strings.
        Gaussian kernels, defined as $K_{\sigma}(\bm{x}, \bm{y}) \coloneqq \exp(-\| \bm{x} - \bm{y} \|_2^2 / 2\sigma)$, make good candidates for \mmd\ training, since they are characteristic on the space of probability distributions. However, other kernels can also be reformulated in terms of expectation values. For a more detailed discussion, see \cref{app:cl_training}, where we also show an example of transforming the polynomial kernel to expectation value-based \mmd.

        Having established that we can express the \mmd\ loss function using Pauli-Z expectation values, let us formally define classically trainable quantum generative models as follows.
        
        \begin{definition}[Classically trainable quantum generative model]
        A quantum generative model with an $n$-qubit input state $\rho_0$ and ansatz $U(\mathbf{w})$ is said to be \textit{classically trainable} if any expectation value of the form $
            \Tr\left[U(\mathbf{w})\rho_0 U(\mathbf{w})^{\dagger} Z_{\bm{i}}  \right]$ can be computed in $\mathrm{poly}(n)$ time with $\mathrm{poly}(n^{-1})$ error, where $Z_{\bm{i}}$ are Pauli-Z strings indexed by the set $\bm{i}$ of size $\ell \in \mathcal{O}(1)$. 
        \label{def:cl_trainable}
        \end{definition}

    \subsection{Fermion Sampling}
        
        Fermion Sampling is a recently proposed scheme for demonstrating quantum computational advantage, based on FLO with magic input states~\cite{oszmaniec2022fermion}.
        For more details on fermionic systems, see \cref{app:fermionic}.
        
        While Gaussian fermionic states evolved under FLO circuits can be efficiently simulated on a classical computer, the use of non-Gaussian resources---so-called magic states---renders the sampling problem classically intractable under standard complexity-theoretic assumptions.
        
        In the Fermion Sampling scheme, one considers $d=4N$ fermionic modes initialized in a tensor product of identical four-mode magic states, i.e.,
        \begin{equation}\label{eq:fs_in}
            \ket{\Psi_{\mathrm{in}}} \coloneqq \ket{\psi_{\text{in}}}^{\otimes N}, \qquad
            \ket{\psi_{\text{in}}} \coloneqq \frac{1}{\sqrt{2}} \left( \ket{1100} + \ket{0011} \right),
        \end{equation}
        which can be efficiently prepared on a quantum processor.
        In the Fermion Sampling scheme, a Haar-random FLO transformation $\uflo$ is applied to this state, followed by fermionic particle-number measurements. 
        The probability of the output configuration $\bm{x}$ is
        \begin{equation}\label{eq:fermion_sampling_prob_dist}
            q(\bm{x}) = \left| \bra{\bm{x}} \uflo \ket{\Psi_{\mathrm{in}}} \right|^2 ,
        \end{equation}
        where the bitstring $\bm{x}$ has an even Hamming weight $|\bm{x}|$. 
        The task of sampling from the probability distribution 
        $q(\bm{x})$ given a random FLO circuit $\uflo$ is referred to as the Fermion Sampling problem.
        
        The FLO transformations are evolutions under quadratic fermionic Hamiltonians and correspond to projective unitary representations of the special orthogonal group $\SO(2d)$. 
        Considering Gaussian input states, output probabilities reduce to determinants and are efficiently computable on a classical computer, making such circuits classically tractable to sample from.
        However, when initialized with the above non-Gaussian product state (which admits high Gaussian extent), the output amplitudes are given by mixed discriminants---quantities that are \#P-hard to compute and generalize the matrix permanent. 
        This introduces computational hardness analogous to that of Boson Sampling, but within the fermionic setting. The hardness results can be summarized as the following:
        
        \begin{proposition}[Informal version of Theorem 3 from Ref.~\cite{oszmaniec2022fermion}]
        \label{lemma:fermion_sampling}
            Under reasonable complexity theoretic assumptions, sampling from the probability distribution $q(\bm{x}) = \left| \bra{\bm{x}} \uflo \ket{\Psi_{\mathrm{in}}} \right|^2$ is intractable by a classical computer, where $\uflo$ is a unitary FLO transformation choosen from the Haar distribution, and 
            \begin{equation} 
            \ket{\Psi_{\mathrm{in}}} = \ket{\psi_{\text{in}}}^{\otimes N}, \qquad
            \ket{\psi_{\text{in}}} = \frac{1}{\sqrt{2}} \left( \ket{1100} + \ket{0011} \right).
            \end{equation}
        \end{proposition}
        Thus, Fermion Sampling provides a conceptually simple and potentially robust route to quantum advantage. 
        Its reliance on FLO circuits makes it naturally compatible with existing qubit-based quantum computer architectures, where matchgate-type operations and fermionic encodings can be realized efficiently.

    \subsection{Related work}
    The idea of training quantum generative models classically was first explored in Ref.~\cite{kasture2023protocols}, however, their framework relies on estimating global probabilities, limiting the scalability of the approach. Building on the results from Ref.~\cite{rudolph2024trainability}, Recio-Armengol et al. proposed a classical training scheme based on the expectation value-based \mmd\ loss function \cite{recio2025iqpopt, recio2025train}. All of these works considered the class of parametrized IQP circuits. Subsequently, in Ref.~\cite{kurkin2025universality} the authors proved universality of the output probability distribution of these parametrized IQP circuits, having access to additional resource qubits and an exponential number of gates. Another recent work~\cite{huang2025generative} discusses and demonstrates classical training of a different class of quantum generative models based on 2D shallow quantum neural networks.

    On the classical simulation side, our work can be related to operator backpropagation methods, such as the Pauli or Majorana Propagation algorithms~\cite{pauliprop, miller2025simulation}. While this type of Heisenberg evolution can be implemented efficiently for the Fermion Sampling circuits considered in this work, the time complexity of our forward evolution has a lower order polynomial scaling, than Heisenberg evolution for Majorana operators.

\section{Fermionic models for generative learning}
\label{sec:fbms}

In this section, we introduce our class of quantum generative models, called Fermionic Born Machines, describing their key characteristics: potential for sampling hardness, and classical trainability. Here we concentrate on the main results and defer the proofs, as well as the explicit simulation algorithm, to \cref{app:training_strategy}.

\subsection{Model construction}

Having established the basics of classical training and Fermion Sampling, we now introduce our framework of \textit{Fermionic Born Machines} (FBMs), which are described by a parametrized FLO ansatz $\uflo(\bm \theta)$ over $d$ fermionic modes and a parametrized input state $\rho(\bm{\alpha})$. In its most general form, we define this model class as follows:
\begin{definition}[Fermionic Born Machine]
\label{def:fbm}
    A FBM on $d = N(k+m)$ fermionic modes is a quantum circuit with the following steps:
    \begin{enumerate}
        \item Initialize each of the $N$ registers of $k+m$  modes each in some (parametrized) quantum state $\rho(\bm{\alpha})$;
        \item Apply parametrized FLO transformation $\uflo(\bm{\theta})$;
        \item For each of the $N$ registers, perform particle number measurement on the first $k$ modes.
    \end{enumerate}
\end{definition}

Based on the numbers $N, k, m$ and the input state $\rho(\bm{\alpha})$, these models have different expressive power and simulability. For example, if $\rho(\bm{\alpha})$ is Gaussian, particle-number measurements can be simulated efficiently on a classical computer. We refer to such models as free-FBMs. In this setting, the input state $\rho$ is fixed, and only the parameters of the FLO ansatz are optimized.

We can significantly improve this fixed-input Gaussian case by considering a more flexible magic input state that increases the number of trainable parameters while also providing non-Gaussian resources. We define this parametrized input state as $\rho(\bm{\alpha}) = \ketbra{\bm{\alpha}}$, where
\begin{equation}
\label{eq:input_state}
    \ket{\bm{\alpha}} \coloneqq \bigotimes_{j=1}^{N} \ket{\alpha_j}, \qquad \ket{\alpha_j} \coloneqq \cos \alpha_j \ket{0000}+\sin \alpha_j\ket{1111}
\end{equation} and $\alpha_j \in [0, 2\pi)$ are trainable parameters. Beyond introducing a number of free parameters, the corresponding subclass of FBMs also contains classically hard cases, as summarized in the following observation.
\begin{observation}[FBM  sampling hardness with parametrized magic inputs]
    The output probability distributions of FBMs over $d=4N$ fermionic modes with parametrized
    input states $\rho(\bm{\alpha}) = \ketbra{\bm{\alpha}}$, as in \cref{eq:input_state}, contain cases where the classical sampling is hard, under the same complexity theoretic assumptions as for Fermion Sampling.
\end{observation}
\begin{proof}
    In \cref{lemma:gaussian_equivalence} from \cref{app:input_state} we show that any $4$-mode, even-parity fermionic Gaussian state can be prepared by applying a Gaussian transformation to such a parametrized state. The fact that the Fermion Sampling input state, as in \cref{lemma:fermion_sampling}, is the product of $4$-mode, even-parity fermionic Gaussian states completes our proof.
    In particular, choosing $\bm{\alpha} = (\pi/4, \dots, \pi/4)$ makes $\ket{\bm{\alpha}}$ equivalent to $\ket{\bm{\Psi}_{\text{in}}}$ from \cref{eq:fs_in} up to a Gaussian transformation.
\end{proof}

Now, while this setting leads to provable sampling hardness, it also restricts the model probability distribution to the even subspace. While this can still be useful for some specific use-cases, in general, it is an unwanted restriction.
We tackle this by choosing to measure only the first $3$ modes in each of the $N$ registers, aiming to model the target probability distribution with the marginal over these modes, i.e., we set $k=3$, $m=1$. This construction takes a step away from the initial concept of QCBMs in the sense that only a subset of modes is measured in the end, the rest is used as a resource, similar to Boltzmann Machines~\cite{ackley1985learning, murphy2012machine}, and, more recently, IQP-QCBMs~\cite{kurkin2025universality}. These hyperparameter settings balance classical sampling hardness and expressivity, and we will rely on these settings in our numerical experiments.

\subsection{Classical training}

Having defined the input state $\rho(\bm{\alpha}) = \ketbra{\bm{\alpha}}$, we now turn to the simulability of expectation values and show that such FBMs are classically trainable.

\begin{observation}[Classical trainability of \fbm s]
\label{thm:main}
FBMs over $d=4N$ fermionic modes, as in \cref{def:fbm}, that implement FLO transformations $\uflo$ on parametrized input states $\ketbra{\bm{\alpha}}$, where $\ket{\bm{\alpha}} \coloneqq \bigotimes_{j=1}^{N}(\cos \alpha_j \ket{0000}+\sin \alpha_j\ket{1111})$ are classically trainable quantum generative models, as in \cref{def:cl_trainable}. Moreover, expectation values of the form $\bra{\bm{\alpha}}\uflo(\bm\theta)^{\dagger} Z_{\bm{i}} \uflo(\bm\theta) \ket{\bm{\alpha}}$ can be computed in $\mathcal{O}(\ell^3 4^{\ell} N^{\lfloor \ell/2 \rfloor})$ time, where the index set $\bm{i}$ is of length $\ell \in \mathcal{O}(1)$.
\end{observation}
For a proof of the time complexity, refer to \cref{lemma:simulation_complexity} in \cref{app:cl_training}. The algorithm relies on decomposing the local components of $\ketbra{\boldsymbol{\alpha}}$ into four Gaussian operators. The Gaussian operator associated with the covariance matrix of $\ketbra{\boldsymbol{\alpha}}$ correctly reproduces the expectation values of all length-$0$ and length-$1$ $Z$-strings. To handle longer $Z$-strings, one replaces the corresponding components of this operator by the four Gaussian operators in the decomposition.
For comparison, we can consider the Heisenberg evolution of Majorana observables. In this picture, naively evolving an $\ell$-long Pauli-Z string under an FLO transformation requires $\mathcal{O}(\ell^3 16^\ell N^{2\ell})$ operations due to the number of possible Majorana strings of the same length. For the explicit algorithm for estimating expectation values and our approach to estimate the \mmd\ loss, please refer to \cref{app:Z_expvals} and \ref{app:estimating_mmd}, respectively.

Using the backpropagation algorithm, we can also calculate the gradients efficiently, enabling the classical optimization of the parameters $\bm{\alpha}$ and $\bm{\theta}$. In our implementation, we used JAX~\cite{jax2018github} for automatic differentiation. The source code is available on GitHub; see Ref.~\cite{fbm_classical_training}.

While classical training overcomes the problems associated with experimentally evaluating gradients, 
barren plateaus can still pose a significant bottleneck for optimization. However, as shown in Refs.~\cite{cerezo2025does, diaz2023showcasing}, for fixed input states, the subsequent Gaussian transformation $\uflo(\bm{\theta})$ does not exhibit barren plateaus in the parameters $\bm{\theta}$, if the Pauli-observable is of constant locality. In order to show that there is no barren plateau in the $\bm{\alpha}$ parameters, we refer to \cref{app:barren_plateau}.
Consequently, the absence of barren plateaus in the circuit parameters implies, via the chain rule, that the loss function should not exhibit exponentially vanishing gradients.

\subsection{Sampling mode}
    In order to sample from the classically trained model, we need to implement it on an actual quantum device. This amounts to preparing the parametrized input state $\ket{\bm{\alpha}}$ and the FLO transformation.

    First, we know that the input state can be prepared using $N$ single-qubit $R_y$ rotations and $3N$ \textsc{CNOT} gates by writing
    \begin{equation}
        \ket{\alpha} = \cos \alpha  \ket{0000} + \sin \alpha  \ket{1111} =  \mathrm{CNOT}_{34} \mathrm{CNOT}_{23} \mathrm{CNOT}_{12}R_y(2 \alpha) \ket{0000}.
    \end{equation}
    Assuming parallel implementation, this can be achieved in constant depth. 
    Second, we can implement any FLO transformation as described in Appendix A of Ref.~\cite{oszmaniec2022fermion} using $\mathcal{O}(d^2)$ matchgates consisting of $R_{z}$ and $R_{xx}$ rotations. Finally, a sample is simply drawn by performing computational basis measurement on the given qubits. Therefore, the complexity of producing a single sample has at most quadratic cost in the system size.

\section{Numerical experiments}
\label{sec:demonstrations}

In this section, we investigate the performance and applicability of our fermionic model and training framework, also benchmarking it against classical models. 
After describing the models and their relevant tunable hyperparameters, we present our first two experiments that aim to demonstrate important aspects of expectation value-based \mmd\ training and FBMs. Subsequently, in the latter two demonstrations, we benchmark these FBMs on relevant generative learning problems, as formulated in the following:

\begin{problem}[Generative learning]
    Given a dataset $\mathcal{X}$ sampled from a target probability distribution $\p$ over $n$ binary random variables, a metric $d$, and a small constant $\varepsilon$, output a representation of a model probability distribution $q_{\mathbf{w}}$ satisfying $d(\p, \q) \leq \varepsilon$.
\end{problem}

To investigate model performance, we rely on different metrics. for small-scale problems, where the model probability distribution can be computed exactly, we study the convergence to the target distribution as measured in the TVD. For larger problems, we first study the similarity of the covariance matrices of the test set and the different trained models, as defined in \cref{app:model_eval}. We further investigate \mmd\ with respect to the test set using a range of $\sigma$ bandwidths. Although implementing unbiased benchmarks would require sampling from the trained models---which is not feasible for FBMs without a quantum device---these measures can still serve as useful initial indicators of model performance.

Also note that in these experiments, the \mmd\ loss function is approximated using a finite number of Pauli-Z strings sampled from the probability distribution defined by the truncated Gaussian kernel function. However, the expectation values themselves are computed exactly without relying on additional statistical estimators.

Parameter initialization is a critical feature of such training frameworks that can significantly accelerate the optimization process. In this work, we rely on random initialization, but note that this represents a worst-case scenario, leaving room for improvement on several fronts.

\subsection{Benchmark models}
To assess the performance of our proposed FBM construction, we provide comparisons against a set of classical benchmark models that represent different levels of statistical and representational complexity.

\subsubsection{Empirical Chow-Liu tree approximation}
As a first statistical approximation of the dataset, we construct a \textit{Bayesian network} that aims to match all univariate marginals and the most important pairwise marginals of the dataset. This is done using the Chow-Liu algorithm~\cite{chow1968approximating}, as implemented in the \texttt{pgmpy} Python library~\cite{Ankan2024}.
Firstly, the algorithm computes the mutual information of the random variables $I(x_i,x_j)$, then constructs the maximum weight spanning tree. Secondly, the direction of the tree is determined by the root node, which is simply chosen as the first random variable. Finally, the parameters of this network are computed using maximum likelihood estimation. Since this is a tree Bayesian network, each node only has a single parent; therefore, exact sampling from the joint probability distribution can be performed efficiently locally.
In the evaluation phase, the goodness of this approximation is given by the \mmd\ distance between the newly sampled dataset and the original test set. Since this is a statistical model, the only relevant hyperparameter is the number of test samples that we always set to a large constant.

Note that although this is an efficient statistical strategy, it only works well for small problems, breaking down for larger numbers of random variables that do not follow a tree-like correlation structure.

\subsubsection{Restricted Boltzmann Machines}
 Restricted Boltzmann Machines (RBMs)~\cite{ackley1985learning, murphy2012machine} aim to model the joint probability distribution over $n$ binary random variables using a pairwise Markov network composed of $n$ visible and $m$ hidden units. 
The corresponding bipartite graph is described by a weight matrix associated with the connections between visible and hidden units. Given the weight matrix $W$ along with local bias vectors $\bm{a}$ and $\bm{b}$, the associated energy function can be written as
\begin{equation}
    E(\bm{v},\bm{h}) = -\bm{a}^{T}\bm{v} -\bm{b}^{T}\bm{h} - \bm{v}W\bm{h},
\end{equation}
where $\bm{v}$ and $\bm{h}$ are vectors of the state of visible and hidden units (bits), respectively.
The joint probability distribution is the corresponding Boltzmann distribution 
\begin{equation}
    q(\bm{v},\bm{h}) =  \frac{e^{-E(\bm{v},\bm{h})}}{Z},
\end{equation}
where $Z$ is the partition function. Finally, the marginal probability distribution over visible units can be obtained by summing over all hidden states
\begin{equation}
    q(\bm{v}) =\frac{1}{Z} \sum_{\bm{h}}  e^{-E(\bm{v},\bm{h})}.
\end{equation}
Similarly to Ref.~\cite{recio2025train}, we also tune the parameters of these RBM models using persistent contrastive divergence, as implemented in scikit-learn's \texttt{BernoulliRBM} class~\cite{scikit-learn, qml_benchmarks}. As a relevant hyperparameter, we consider the number of hidden units in the range $m \in [\lfloor n/2 \rfloor, n]$ and select the one with the best \mmd\ scores over a range of $\sigma$ bandwidths.

\subsubsection{\fbm\ with magic states}
As introduced in the previous section, our proposed model is a \fbm\ on $d=4n/3$ fermionic modes, where the registers are defined by the hyperparameters $k=3$, $m=1$, and $N = n/3$. For simplicity, we assume that $n$ is a multiple of $3$. This model consists of the magic input state with trainable parameters $\ket{\bm{\alpha}} \coloneqq \bigotimes_{j=1}^{N}(\cos \alpha_j \ket{0000}+\sin \alpha_j\ket{1111})$ and the parametrized transformation $\uflo (\bm{\theta})$. The resulting generative model has $N+d(d+1)/2$ free parameters. We also consider settings in which multiple FLO layers are applied to the input state, introducing an additional hyperparameter: the number of layers. Other relevant training hyperparameters include the learning rate and the  number of training iterations. In our implementation, we use the Adam optimizer~\cite{adam2014method}, but the algorithm can be easily adapted to other gradient-based optimizers. To estimate the \mmd\ loss function, we use multiple $\sigma$ bandwidths and a fixed locality cutoff for the Pauli-Z strings. In our numerical experiments, we will refer to these models simply as \fbm s.

\subsubsection{Free-FBMs}
To showcase the power of magic in the context of \fbm s, we also consider the case of fixed $\bm{\alpha} = \bm{0}$, corresponding to a fixed Gaussian input state $\ket{\bm{0}} = \bigotimes_{i=0}^n \ket{0}$, followed by a parametrized FLO transformation $\uflo (\bm{\theta})$. This variant removes $N$ trainable parameters and restricts the model to Gaussian states, where efficient strong simulation is possible~\cite{terhal2002classical}. The relevant hyperparameters are the same as in the previous case. We refer to these models as \textit{free-FBMs}.

\subsection{Small-scale Markov network benchmark}

\begin{figure}[t]
    \centering
    \includegraphics[width=0.8\linewidth]{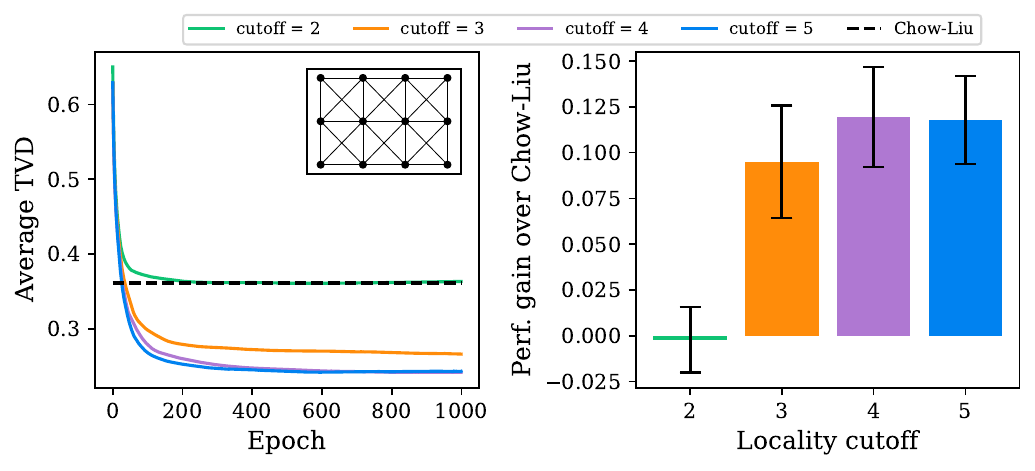}
    \caption{\textbf{Training performance of small FBMs on graph-structured problems.} FBMs are trained on $10$ datasets sampled from different Markov networks of the same grid topology (shown in left inset). The training \mmd\ loss is approximated using different cutoffs. (left) The mean total variation distance (TVD) is shown along the training. (right) The final TVD improvement over the Chow-Liu approximation of the dataset is presented.}
    \label{fig:mrf_exp}
\end{figure}
We start our numerical investigations by considering small-scale problems, where the explicit model probability distributions can be constructed, and consequently, explicit distance metrics, such as the TVD, can be computed. Here, our goal is two-fold, we seek to investigate whether training our model on the \mmd\ loss function shows any convergence in the TVD, and we also analyze the performance dependence on the maximal Z-string locality.

In this experiment, we consider a structured problem, based on a $12$-node Markov network (MN) with a grid topology of $3\times4$, as shown in the inset of \cref{fig:mrf_exp} (left). We use the benchmark proposal from Ref.~\cite{bako2024problem}, and assign random factor values to each of the $6$ maximal cliques of size $4$. The target probability distribution is then calculated by taking the factor product of these and normalizing with the partition function. Finally, the training set is constructed by sampling from this probability distribution $1000$ times. To construct the explicit model probability distribution, we use the fermionic backend implemented in the Piquasso software package \cite{kolarovszki2025piquasso}.

In each training iteration, we compute all possible Z-string expectation values up to a given locality $\ell$ from both the data, and the training set. The loss function is then computed as the sum of the squared differences, weighted by the probability distribution that corresponds to the Gaussian kernel with standard deviation $\sigma = 1.0$.

Due to the size of the problem, we can analyze the training performance, as captured by the TVD between the model probability distribution and the target.  Our goal is to compare the performance of FBMs when the loss function is cut off at different Z-string localities, and we also compare the performance to the Chow-Liu tree approximation of the dataset.
Taking the average over $10$ independent problems with the same MN structure, we show the results in \cref{fig:mrf_exp} (left). We also compare the final performance to the TVD achievable with the Chow-Liu approximation of the dataset in \cref{fig:mrf_exp} (right), where the error bars show the standard deviation over the $10$ independent problem instances. Our results show that although the model is trained only to match low-order moments of the target distribution, this objective nonetheless yields convergence in terms of TVD as well. Furthermore, increasing the locality cutoff from $2$ to $3$, and $4$ improves the performance, however, including length-$5$ Z-strings does not improve it further. This observation is related to the factorization of positive MNs, since those are completely characterized by their moments up to order $\ell$, where $\ell$ is the size of the largest clique~\cite{barndorff2014information}. Therefore, in the case of such problems, training our models on Z-string expectation values up to locality $\ell$ can be sufficient. For a detailed discussion of the effect of locality cutoff in general, please refer to Ref.~\cite{herrero2025born}.

\subsection{Effect of overparametrization}

\begin{figure}[t]
    \centering
    \includegraphics[width=\linewidth]{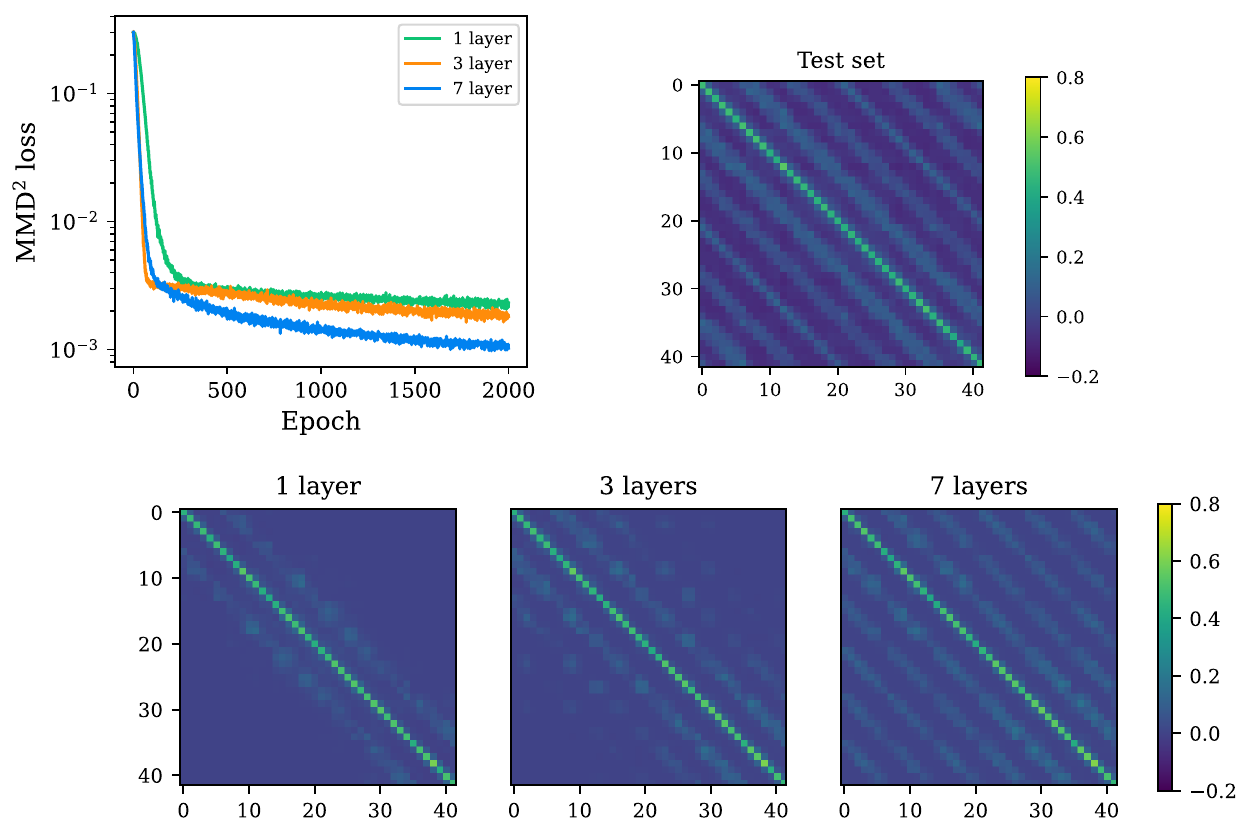}
    \caption{\textbf{The effect of overparametrization on the performance FBMs.} The models were trained on the equilibrium states of a cellular automaton over a $6\times 7$ grid. FBMs were constructed with an increasing number of layers, and we trained an RBM with $40$ hidden units for comparison.}
    \label{fig:ca_experiment}
\end{figure}

In this experiment, we consider a larger problem over $42$ binary random variables, to investigate the overparametrization of FBMs (with $56$ modes). The training dataset is composed of the equilibrium states of a \textit{cellular automaton} (CA) on a $6\times7$ grid using Conway's Game of Life rules~\cite{adamatzky2010game}. To obtain one sample, we start the CA in a random state and let it evolve for $200$ steps, by which time it has reached an equilibrium state. We repeat this multiple times, disregarding all-zero states, to collect the training and test sets, both of size $1000$.

We train our fermionic models on only one- and two-long Z-string expectation values to investigate whether they can reproduce the covariance matrix of the test set, as shown in \cref{fig:ca_experiment} (top right). We train FBMs with a learning rate of $0.001$ and Gaussian kernel bandwidth $\sigma = 0.1$. 

First we train the FBM with a single layer and find that the final covariance matrix is close to being diagonal, indicating that in the model probability distribution, distant random variables appear conditionally independent of each other (zero covariance). Interestingly, adding multiple FLO layers with independent trainable parameters substantially improves performance, yielding lower final loss values (\cref{fig:ca_experiment}, top left) and covariance matrices of noticeably closer to the test set (\cref{fig:ca_experiment}, bottom). This shows the benefit of overparametrizing the FLO ansatz, and importantly, while this increases the number of trainable parameters, the trained model can be compiled to a single layer corresponding to a single FLO transformation. Hence, overparametrizing the model does not come with any inference time overhead on the quantum device. Although these observations can guide our model construction, they may also be of independent theoretical interest.

These numerical results are also consistent with the theory of overparametrization developed for quantum circuits, where the extended number of parameters improve the loss landscape~\cite{Larocca_2023}. Overparametrization is also studied for the special case of QCBMs~\cite{delgado2023identifyingoverparameterizationquantumcircuit}, where the authors observed similar phase transitions in the number of parameters, and observed that gradient-based training can become highly efficient after a critical circuit depth.

\subsection{Molecular dataset}
\begin{figure}[t]
    \centering
    \includegraphics[width=\linewidth]{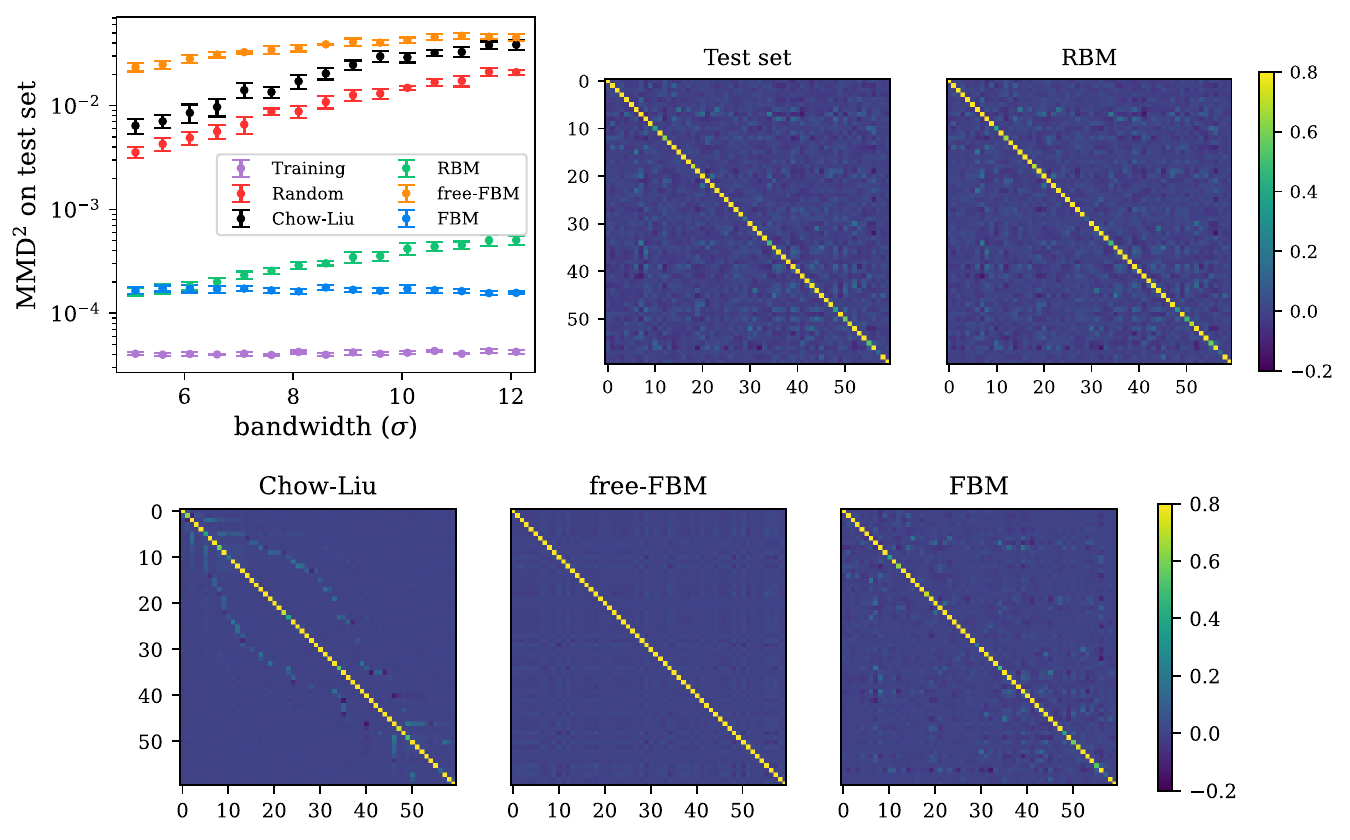}
    \caption{\textbf{Test results on the molecular fingerprint dataset.} The models were trained on $60$-bit long molecular fingerprints. FBM have $7$ trainable layers, and we trained an RBM with $44$ hidden units for comparison. The FBM with magic input states captures the structure present in the dataset.}
    \label{fig:molecular_experiment}
\end{figure}

As a first practical demonstration, we consider a molecular dataset serving as a structured benchmark.
Molecule data was obtained from the ZINC20 database~\cite{irwin2020zinc20}, a publicly available collection of commercially accessible compounds for virtual screening.
We sampled $10^5$ random molecules with their structure represented as SMILES strings, which were subsequently converted to fixed-length bitstrings using Extended-Connectivity Fingerprints (ECFP, also known as Morgan fingerprints)~\cite{rogers2010extended}, as implemented in the RDKit cheminformatics toolkit~\cite{rdkit2025}. These fingerprints were computed with a radius $2$ and bit length $60$, where each bit encodes the presence or absence of a chemical substructure. The resulting bitstring dataset was finally divided into equal-sized training and test sets.

For this experiment, we trained \fbm s with $80$ modes and $7$ layers, for $2000$ epochs and with a $0.001$ learning rate. The loss function was calculated as the mean of two \mmd\ losses with different bandwidths derived from the median heuristic~\cite{garreau2017large, recio2025train}: $5.099$, and $10.198$. We estimate the loss function up to $5$-long Z-strings. For comparison, we also train a free-FBM with the same hyperparameters and an RBM with $44$ hidden units (after fine-tuning) for $1000$ epochs with $0.01$ learning rate. 

The results, as reflected in the \mmd\ with respect to the test set, and the final covariance matrices are shown in \cref{fig:molecular_experiment}. Being a larger problem that potentially lacks the tree correlation structure, the Chow-Liu approximation is further from the test set than a dataset sampled completely at random. Furthermore, free-FBMs perform even worse, highlighting that the magic input states introduced in this work not only enable smapling hardness, but significantly increase model expressivity. Our FBM model with magic input states manages to capture the structure of the dataset with small errors, producing lower test \mmd\ than the corresponding RBM for large $\sigma$ bandwidths. This behavior is related to the fact that the trained FBM captures single-variable marginals better than the RBM.

\subsection{Gene sequence dataset}

\begin{figure}[t]
    \centering
    \includegraphics[width=\linewidth]{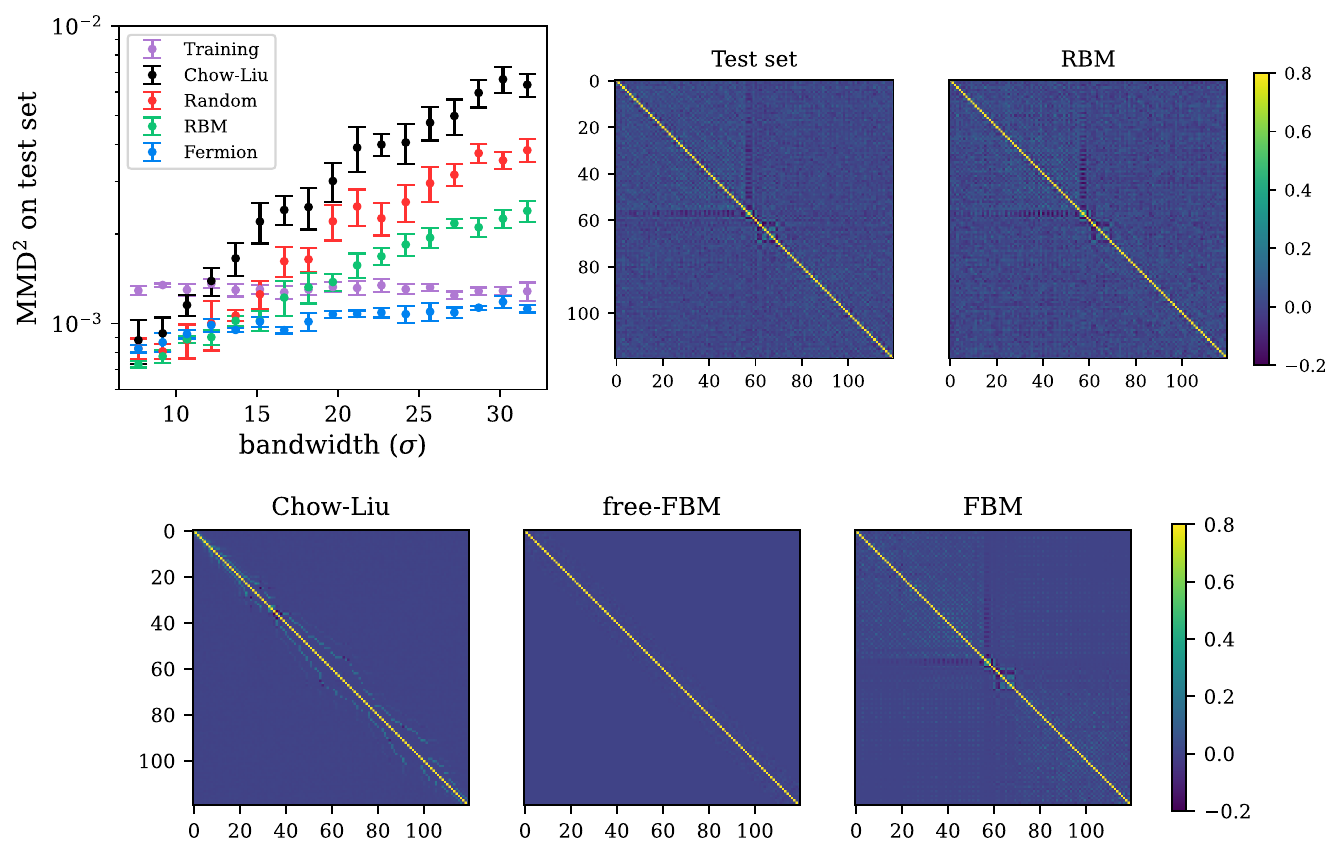}
    \caption{\textbf{Test results on the gene sequence dataset.} The models were trained on $120$-bit long bitstrings obtained from $60$-long gene sequences. The used FBMs have $10$ trainable layers, and we trained an RBM with $62$ hidden units for comparison.}
    \label{fig:gene_experiment}
\end{figure}

Finally, we consider the splice-junction gene dataset composed of $60$-long gene sequences~\cite{splice-junction}. We select the samples with no ambiguity, meaning, that each gene is a symbol from the set $\{A, G, T,C\}$, and these are then transformed to bitstrings using binary encoding $\{00, 01, 10, 11\}$. Consequently, the final dataset consists of $3175$ bitstrings of length $120$, which we divide into a training set of $1270$, and a test set of $1905$. The corresponding FBM has $160$ modes.

After some fine tuning, the FBMs have $10$ FLO layers and are trained for $1000$ epochs with $0.001$ learning rate. Here we also rely on the median heuristic~\cite{garreau2017large} for choosing two $\sigma$ bandwidths: $7.681$ and $15.362$. The fine-tuned RBM model has $120$ visible and $62$ hidden units.

The test \mmd\ and the final covariance matrices are shown in \cref{fig:gene_experiment}, where we omitted the \mmd\ of the free-FBM, since it produces significantly weaker results. For low bandwidths, all models produce the same performance, meaning that higher-order correlations are indistinguishable from a uniformly random distribution. Furthermore, these all have lower \mmd\ than the training set. All these observations can be attributed to the low number of samples in the training and test set. As we go to higher bandwidths, where low-body correlations dominate, the difference between the models becomes more apparent.

\section{Conclusion and outlook}
\label{sec:conclusion}

In this work, we introduced Fermionic Born Machines, a class of quantum generative models that allow classical training using the \mmd\ loss function. Based on the Fermion Sampling quantum advantage scheme, we have shown that using parametrized magic input states, these models can express probability distributions that are classically hard to sample from. By decomposing the input states, we derived a classical algorithm for computing the expectation values of Pauli-Z strings, which operates in $\mathcal{O}(\ell^3 4^{\ell} N^{\lfloor\ell/2\rfloor})$ time, where $N$ is the number of $4$-qubit registers and $\ell$ is the length of the Z-string. This specialized approach outperforms operator backpropagation methods for our specific problem. 

In our numerical experiments, we demonstrate that, for structured problems, training on low-body observables yields good performance. Furthermore, we investigated the overparametrization of the quantum circuit by implementing several subsequent FLO layers and we have seen that while a single layer is enough to parametrize the whole subspace, this type of overparametrization can be highly beneficial in training such FBMs. These observations are connected to information spreading in FLO circuits and may be of independent interest. The subsequent experiments considered practically relevant datasets based on molecular fingerprints and gene sequences, demonstrating the potential of FBMs. 

In order to take a step towards demonstrating useful quantum advantage in machine learning, the target problem selection is just as important as model construction. 
Since through this classical training framework based on Gaussian \mmd\ loss with non-vanishing bandwidths, 
the model learns to approximate (informationally) local correlations present in the data, the target probability distribution should also be structured and representable with local correlations. 
In structured problems, local correlations can be efficiently represented using probabilistic graphical models, such as Markov networks (MNs)~\cite{koller2009probabilistic}. In the edge case of pairwise MNs, i.e., those that can be described by $2$-local correlations, the training can be performed considering only up to $2$-local expectation values. Furthermore, beyond pairwise correlations, classical parameter estimation and sampling are both computationally hard tasks, leaving room for quantum advantage~\cite{bako2024problem}.

Our study can be extended also on several other fronts.
First, in our setting, classical training remains tractable only for Z-strings of length $\mathcal{O}(1)$. By contrast, the IQP framework~\cite{recio2025iqpopt} does not have this limitation, enabling the approximate simulation of longer Z-strings. However, obtaining an approximation algorithm for FBMs, with better scaling, would potentially push practical simulability to thousands of qubits and higher-order correlations.
Furthermore, we aimed to demonstrate the training of FBMs, but, in order to perform fair benchmarks against established classical models, these models require further hyperparameter tuning and final sampling on a quantum device. We leave these considerations for future work.

Finally, beyond enabling fully classical training of such structured quantum generative models, one can also leverage the classical training to warm-start more general models on a quantum device.
This can be envisioned in two scenarios. For a problem with $n$ binary random variables, one can break it up into $k$ smaller problems such that $n = \sum_{i=1}^k n_k$ and after considering the marginal datasets, we train the smaller models separately. In the final training, we concatenate the pre-optimized circuits and implement additional trainable gates that connect these regions. In the second setting, the model is not broken up into smaller sub-models, but trained classically, up to $\ell$-body correlations, learning higher-order correlations on a quantum device, where sampling is possible.

\section{Acknowledgement}

We thank Joseph Bowles and Michał Oszmaniec for fruitful discussions. The authors would also like to thank the support of the Hungarian National Research, Development and Innovation Office (NKFIH) through the KDP-2023 funding scheme, the Quantum Information National Laboratory of Hungary and the grants TKP-2021-NVA-04, TKP2021-NVA-29 and FK 135220. ZZ was partially supported by the Horizon Europe programme HORIZON-CL4-2022-QUANTUM-01-SGA via the project 101113946 OpenSuperQPlus100 and the QuantERA II project HQCC-101017733.
The authors also acknowledge the computational resources provided by the Wigner Scientific Computational Laboratory (WSCLAB).

\printbibliography[heading=bibintoc]

\appendix

\section{Training quantum generative models}
    \label{app:cl_training}
    This section provides a review of training methods for QCBMs, concentrating on loss function formulations and the classical training paradigm. In each case, the probability distributions are defined over $n$ binary random variables, but a similar treatment is possible for higher order discrete or contiuous cases.
    \subsection{Training on a quantum device}
        The discrepancy between the target $\p$ and model probability distribution $\q$ can be characterized using distance measures, such as the \textit{total variation distance} (TVD) defined as
        \begin{equation}
            \mathcal{L}_{\rm{TVD}}(\p,\q) \coloneqq \frac{1}{2} \sum_{\x \in \{0,1\}^n} |\p(\x) - \q(\x)|,
        \end{equation}
        or the \textit{Kullback-Leibler (KL) divergence} defined as
        \begin{equation}
            \mathcal{L}_{\mathrm{KL}} (\p, \q) \coloneqq \sum_{\x \in \{0,1\}^n} \p(\x) \log\left(\frac{\p(\x)}{\q(\x)}\right).
        \label{eq:kldiv}
        \end{equation}
        However, since these are explicit loss functions, that is, they rely on explicit probabilities, they are not suitable for training implicit learning models, such as QCBMs~\cite{rudolph2024trainability}. A more suitable loss function is the squared maximum mean discrepancy (\mmd) defined using the training samples obtained from $\p$ and the samples obtained from the model distribution $\q$ as
        \begin{equation}
            \mathcal{L}_{\mathrm{MMD}^2}(\p, \q) \coloneqq \mathds{E}_{\x, \y \sim \p}[K(\x, \y)] -2 \mathds{E}_{\x \sim \q, \y \sim \p}[K(\x, \y)] + \mathds{E}_{\x, \y \sim \q}[K(\x, \y)],
            \label{eq:mmd}
        \end{equation}
        where the kernel function $K$ is usually chosen as a Gaussian kernel
        \begin{equation}
            K_{\sigma}(\x, \y) \coloneqq e^{-\frac{\| \x-\y\|_2^2}{2\sigma}} = \prod_{i=1}^{n} e^{-\frac{(x_i-y_i)^2}{2\sigma}},
        \label{eq:kernel}
        \end{equation}
    with $\sigma$ being the standard deviation or bandwidth.
    \subsection{Classical training with Gaussian \mmd}
        Using the \mmd\ loss function in the form of \cref{eq:mmd} requires sampling the model during training. However, in Ref.~\cite{rudolph2024trainability} it was shown that it can be reformulated with expectation values of Hermitian operators as 
        \begin{equation}
            \mathcal{L}_{\rm{MMD}^2}(\p, \q ) =
            \mathcal{M}(\rho_{\q}, \rho_{\q}) - 2\mathcal{M}(\rho_{\q}, \rho_p) + \mathcal{M}(\rho_p, \rho_p),
        \end{equation}
        where $\rho_{\mathbf{w}} \coloneqq U(\mathbf{w}) \ketbra{0} U(\mathbf{w})^{\dagger}$, $\rho_{p} \coloneqq \sum_{\y} \p(\y)\ketbra{ \y}$ are density operators corresponding to the model and target distributions. $\mathcal{M}$ is defined using the observable
        \begin{equation}
            O^{(\sigma)}_{\rm{MMD}^2} \coloneqq \sum_{\bm{x}, \bm{y} \in \{ 0, 1\}^n } K(\bm{x}, \bm{y}) \ketbra{\bm{x}} \otimes \ketbra{\bm{y}}
        \end{equation}
        by taking its expectation values, i.e.,
        \begin{equation}
            \mathcal{M}(\rho, \rho') \coloneqq 
            \Tr \left[O^{(\sigma)}_{\rm{MMD}^2} 
                (\rho \otimes \rho')
            \right]
            =
            \sum_{\bm{i} \subseteq \mathcal{P}([n])} (1-p_{\sigma})^{n-|\bm{i}|} p_{\sigma}^{|\bm{i}|}\; \mathrm{Tr}\left[ Z_{\bm{i}} \rho \right] \; \mathrm{Tr}\left[ Z_{\bm{i}} \rho' \right],
        \label{eq:mmd_dist}
        \end{equation}
        where $Z_{\bm{i}}$ is the Pauli-$Z$ string indexed by the set $\bm{i}$ with cardinality $|\bm{i}|$ and $p_{\sigma} = (1-e^{-1/(2\sigma)})/2$. 
        Let us denote the probability distribution over the possible Pauli-Z string indices $\bm{i}$ as 
        \begin{equation}
            p_{K_{\sigma}} (\bm{i}) = (1-p_{\sigma})^{n-|\bm{i}|} p_{\sigma}^{|\bm{i}|}.
        \end{equation}
        With this notation, we can write the \mmd\ loss as 
        \begin{equation}
        \label{eq:mmd_expval}
            \mathcal{L}_{\rm{MMD}^2}(p, q_{\mathbf{w}}) = \sum_{\bm{i} \in \mathcal{P}([n])} p_{K_{\sigma}}(\bm{i}) (\expval{Z_{\bm{i}}}_p -\expval{Z_{\bm{i}}}_{q_{\mathbf{w}}})^2 = \underset{{\bm{i}} \sim p_{K_{\sigma}}}{\mathbb{E}}[(\expval{Z_{\bm{i}}}_p -\expval{Z_{\bm{i}}}_{q_{\mathbf{w}}})^2],
        \end{equation}
        where $\expval{Z_{\bm{i}}}_{\p}$ and $\expval{Z_{\bm{i}}}_{\q}$ denote the expectation value of the Pauli-Z strings indexed by ${\bm{i}}$.

        Now it is clear that instead of sampling from $\q$, we can sample indices from the probability distribution defined by the kernel function and compute the expectation values of the corresponding Pauli-$Z$ observables.

    \subsection{Polynomial kernel}
        While the previous derivation of the expectation value-based \mmd\ loss function used Gaussian kernels, the idea can be generalized to other distance measures as well. Another candidate could be the $d$-th order polynomial kernel that while not being characteristic, can be used to train a model to match the target probability distribution up its $k$-th moment (mean, variance, skewness, etc.)~\cite{yang2019polynomial}.
        
        We start by using the following relation:
        \begin{equation}
            \sum_{x = 0}^1 x \ketbra{x} = 0 \ketbra{0} + 1 \ketbra{1} 
            = \frac{1}{2} (I - Z) \eqqcolon \hat{n},
        \end{equation}
        where $\hat{n}$ is the number operator.
        More generally, if we take a polynomial $P(\bm{x}) = P(x_1, \dots, x_n)$ in the binary variables  $\bm{x} = (x_1, \dots, x_n)$, we get
        \begin{equation}
             \sum_{\bm{x} \in \{ 0, 1\}^n } P(\bm{x}) \ketbra{\bm{x}}
             =
            \sum_{x_1, \dots, x_n = 0}^1
            P(x_1, \dots, x_n) \bigotimes_{i=1}^n \ketbra{x_i}
            = p(\hat{n}_1, \dots, \hat{n}_n),
        \end{equation}
        where $p(\hat{n}_1, \dots, \hat{n}_n)$ has the same coefficients as $P(x_1, \dots, x_n)$, and is acquired by replacing $x_i \mapsto \hat{n}_i$ in $P$.
        Let us denote $\zeta(x) = 1-2x$. We can obtain polynomials of $Z$ operators by writing
        \begin{equation}
             \sum_{\bm{x} \in \{ 0, 1\}^n } P(\zeta(\bm{x})) \ketbra{\bm{x}}
             =
            \sum_{x_1, \dots, x_n = 0}^1
            P(\zeta(x_1), \dots, \zeta(x_n)) \bigotimes_{i=1}^n \ketbra{x_i}
            = P(Z_1, \dots, Z_n),
        \end{equation}
        where $\zeta(\bm{x}) = (\zeta(x_1), \dots, \zeta(x_n))$ is understood elementwise.
        Hence, we know that for any polynomial $K(\bm{x}, \bm{y})$ for $\bm{x}, \bm{y} \in \{ 0, 1\}^n$ we get
        \begin{align}
            O_K \coloneqq \sum_{\bm{x}, \bm{y} \in \{ 0, 1\}^n} K(\bm{x}, \bm{y}) \ketbra{\bm{x}} \otimes \ketbra{\bm{y}}
            =
            p(\hat{n}_1, \dots, \hat{n}_n, \hat{n}_{n+1}, \dots, \hat{n}_{n+n}),
        \end{align}
        and similarly,
        \begin{align}
             \sum_{\bm{x}, \bm{y} \in \{ 0, 1\}^n} K(\zeta(\bm{x}), \zeta(\bm{y})) \ketbra{\bm{x}} \otimes \ketbra{\bm{y}}
             = 
             p(Z_1, \dots, Z_n, Z_{n+1}, \dots, Z_{n+n}).
        \end{align}
        For example, choosing $K$ to be the polynomial kernel defined as
        \begin{equation}
            K(\bm{x}, \bm{y}) \coloneqq \left( \zeta(\bm{x})^T \zeta(\bm{y}) + c \right)^d
        \end{equation}
        we simply get
        \begin{equation}
            O_K = \left( \sum_{i=1}^n Z_{i} \otimes Z_{n+i} + c \right)^d.
        \end{equation}
        In the special case of $c=0$ and $d=1$, we get a linear kernel with the corresponding observable
        \begin{equation}
            O_{K} = 
            \sum_{i=1}^n
            Z_k \otimes Z_{n+k}.
        \end{equation}
        This means that if we use it in the expression for the \mmd, we get
        \begin{equation}
            \mathcal{M}(\rho, \rho') = 
            \Tr \left[O_{K} 
                (\rho \otimes \rho')
            \right]
            =
            \sum_{i=1}^n 
             \langle Z_i \rangle_{\rho}  \langle Z_i \rangle_{\rho'}
        \end{equation}
        and hence
        \begin{equation}
            \mathcal{L}_{\mathrm{MMD}^2}(p, q) = 
            \sum_{i=1}^n (\langle Z_i \rangle_{q} - \langle Z_i \rangle_{p})^2.
        \end{equation}

\subsection{Model evaluation}
\label{app:model_eval}
Although for a fair evaluation of the power of the trained model we need samples from the corresponding probability distribution, we can also define metrics that rely on expectation values and therefore can be used as initial performance indicators. Inspired by Ref.~\cite{recio2025train}, we first use the covariance matrix, a quantity that is useful for visualizing the trained expectation values of up to length-2 Z-strings:
\begin{definition}[Covariance]
    The covariance between binary random variables $x_i$, $x_j$ from the joint probability distribution $\p$ is
    \begin{equation}
        \operatorname{cov}_{\p}(x_i, x_j) = \expval{Z_iZ_j}_{\p} - \expval{Z_i}_{\p}\expval{Z_j}_{\p}.
    \end{equation}
\end{definition}
Furthermore, if we have access to additional samples from the target distribution, we can use this as a test set and define the corresponding \mmd\ as follows:
\begin{definition}[\mmd\ to a test set]
    Given a model probability distribution $\q$, the \mmd\ with respect to a test set is
    $
        \mathcal{L}_{\mathrm{MMD}^2} (p^*, \q),
    $
    where $p^*$ is the empirical probability distribution of the test data, and $\mathcal{L}_{\mathrm{MMD}^2}$ is defined in \cref{eq:mmd_expval}.
\end{definition}

\section{A brief review of fermionic quantum systems}\label{app:fermionic} 
    In this section, we review some of the basic notions of fermionic systems required for this work, and their connection to qubit-based circuit architectures. For further details, see, e.g., Ref.~\cite{oszmaniec2022fermion}.

    \subsection{Fermionic Fock space}

    Fermionic quantum systems consist of indistinguishable fermions, which imposes a symmetry constraint on their composite quantum state. If we consider a system of $d$ fermionic modes, the quantum state of a single particle can be modeled by a Hilbert space $\mathcal{H} \coloneqq \mathbb{C}^d$. Multi-particle states consisting of $k$ fermions can be described by $\mathcal{H}^{\otimes k}$, however, the previously mentioned symmetry is not imposed in this space. To remedy this, consider the unitary action \( U_\pi : \mathcal{H}^{\otimes k} \to \mathcal{H}^{\otimes k}  \;\; (\pi \in S_k )\) defined by permuting the particles, i.e., 
    $
        U_\pi \psi = U_\pi \, \phi_1 \otimes \dots \otimes \phi_k \coloneqq \phi_{\pi(1)} \otimes \dots \otimes \phi_{\pi(k)},
    $
    where $\phi_i \in \mathcal{H}$.
    Generally, a vector \( \psi \in \mathcal{H}^{\otimes k} \) is \textit{totally antisymmetric} if for all permutations $\pi \in S_k$ we have
    $
        U_\pi \psi = \sgn(\pi) \psi.
    $
    For fermions, the \(k\)-particle Hilbert space is the subspace of $\mathcal{H}^{\otimes k}$ generated by the totally antisymmetric vectors, and we will denote this by $\mathcal{H}^{\otimes k} \big |_{\text{asymm}}$.
    Consequently, a fermionic system consisting of $d$ fermionic modes is modeled using the \textit{fermionic Fock space} defined as
    \begin{equation}
       \mathcal{F}(\mathcal{H}) \coloneqq \bigoplus\limits_{k=0}^d \; \mathcal{H}^{\otimes k} \big |_{\text{asymm}},
    \end{equation}
    where $\mathcal{H}^{\otimes 0} \big |_{\text{asymm}} \cong \mathbb{C}$.
    We also refer to $\mathcal{H}^{\otimes k} \big |_{\text{asymm}}$ as the \textit{$k$-particle subspace} of the (fermionic) Fock space.

    Let us introduce the following notation for the \textit{antisymmetric tensor product}
    \begin{equation}\label{eqn:vee}
        \phi_1 \vee \dots \vee \phi_n \coloneqq \frac{1}{\sqrt{n!}} \sum_{\pi \in S_n}  \sgn(\pi) \, \phi_{\pi(1)} \otimes \dots \otimes \phi_{\pi(n)},
    \end{equation}
    which expression yields totally antisymmetric vectors by combining $\phi_i \in \mathcal{H}$. Let $\{ e_{j} \}_{j=1}^d \subset \mathcal{H}$ be an orthonormal basis (ONB) in $\mathcal{H}$. The Fock basis states are defined as
    \begin{equation}
        \ket{n_1, \dots,  n_d}
        \coloneqq
        e_{i_1} \vee \cdots \vee e_{i_k}
    \end{equation}
    such that $i_1 < \dots < i_k$, and $n_j$ is $1$ when $j$ can be found among the indices $i_1 < \dots < i_k$, $0$ otherwise. The Fock basis states constitute an ONB in the fermionic Fock space.

    \subsection{Creation, annihilation and Majorana operators}

    For a single-mode fermionic system, the \textit{creation and annihilation operators} are defined by
    \begin{align}
    \begin{split}
        f^\dagger \ket{0} &= \ket{1}, \qquad f^\dagger \ket{1} = 0,\\
        f \ket{0} &= 0, \qquad\phantom{.....} f \ket{1} = \ket{0}.
    \end{split}
    \end{align}
    For multiple fermionic modes, we denote creation and annihilation operators acting on the $i$-th mode as $\fdag_i$ and $\f_i$, respectively. In this setting, we impose the \textit{canonical anticommutation relations}
    \begin{equation}
        \{ f_i, \fdag_j\} = \delta_{i,j} \mathbbm{1}, \qquad \{ f_i, f_j\} = \{ \fdag_i, \fdag_j\} = 0,
    \end{equation}
    and define the action of the creation operators on the multimode vacuum state $\ket{\bm{0}}$ as
    \begin{equation}
        \ket{\bm{x}} \coloneqq (\fdag_1)^{x_1} \cdots (\fdag_{d})^{x_{d}} \ket{\bm{0}}.
    \end{equation}
    This completely defines the creation and annihilation operators. For our purposes, it is useful to introduce another set of operators, called \textit{Majorana operators}, defined by
    \begin{equation}
        m_{2j-1} \coloneqq \f_j + \fdag_j, \phantom{xxxx} m_{2j} \coloneqq -i(\f_j - \fdag_j),
    \end{equation}
    with anticommutation relations $\{m_j, m_k\} = 2 \delta_{j,k} \mathbbm{1}$.

    \subsection{Jordan-Wigner transformation}

    The Jordan-Wigner transformation is a procedure that enables a mapping between fermionic and qubit-based systems. 
    This transformation is modeled by a unitary map $\mathcal{V}_{\text{JW}}$ between the fermionic Fock space $\mathcal{F}(\mathcal{H})$ and $(\mathbb{C}^{2})^{\otimes d}$, and is defined as
    \begin{equation}
        \mathcal{V}_{\text{JW}} \left(
            (f_1^\dagger)^{x_1} \dots (f_d^\dagger)^{x_d} \ket{0}
        \right)
        = \bigotimes_{j=1}^d \ket{x_j} \in (\mathbb{C}^2)^{\otimes d}.
    \end{equation}
    For our purposes, it is easiest to formulate these transformations in terms of Majorana operators as
    \begin{align}
        m_{2j} &\stackrel{\mathcal{V}_{\text{JW}}}{\rightarrow} \mathcal{V}_{\text{JW}}^\dagger m_{2j} \mathcal{V}_{\text{JW}} = Z_0\dots Z_{j-2} X_{j-1},\\ 
        m_{2j+1}&\stackrel{\mathcal{V}_{\text{JW}}}{\rightarrow} \mathcal{V}_{\text{JW}}^\dagger m_{2j+1} \mathcal{V}_{\text{JW}} = Z_0\dots Z_{j-2} Y_{j-1}.
    \end{align}
    Importantly, the reverse transformation of a Pauli-Z operator can be achieved by \begin{equation}
        Z_j \stackrel{\mathcal{V}_{\text{JW}}^{-1}}{\rightarrow}  -i m_{2j-1} m_{2j}. 
    \end{equation}
    Therefore, Z-strings $Z_{\bm{j}} \coloneqq Z_{j_1} \cdots Z_{j_n}$ can be expressed with a product of $2n$ Majorana operators, up to a multiplicative factor:
    \begin{equation}
        Z_{j_1} \cdots Z_{j_n} = (-i)^{n} m_{2j_1-1} m_{2j_1} \dots m_{2j_n-1} m_{2j_n}.
    \end{equation}

    \subsection{Gaussian transformations and Gaussian states}
        In our work, we are primarily interested in \textit{Gaussian} or \textit{fermionic linear-optical} (FLO) transformations\footnote{Sometimes also called \textit{quasi-free} transformations.}. These transformations can be written as $U = \exp(-i H)$, where $H$ is quadratic in the Majorana operators, i.e.,
        \begin{equation}\label{eq:fermionic_quadratic_hamiltonian}
            H = \frac{i}{4}\sum_{j,k}^{2d} A_{j,k} m_jm_k,
        \end{equation}
        where $A = -A^T \in \mathbb{R}^{2d\times2d}$. Operators of this form are called \textit{fermionic quadratic Hamiltonians}.
        It is worth mentioning that Gaussian transformations are closely related to matchgate computation of qubit systems, where several hardness results have been proven~\cite{valiant2002quantum, terhal2002classical, jozsa2008matchgates}.
        The evolution of Majorana operators by a Gaussian transformation $U$ can be written as
        \begin{equation} \label{eq:so_matrix}
            U^\dagger m_i U = \sum_{j=1}^n O_{ij} m_j,
        \end{equation}
        where $O = e^{A} \in \mathrm{SO}(2d)$.

        A fermionic quantum state represented by the density matrix $\rho$ is called a \textit{fermionic Gaussian state} when the density matrix can be written as
        \begin{equation}
            \rho = \frac{e^{-H}}{Z},
        \end{equation}
        where $H$ is a fermionic quadratic Hamiltonian from Eq.~\eqref{eq:fermionic_quadratic_hamiltonian} and $Z \coloneqq \Tr[e^{-H}]$ is a normalization constant. The Fock basis states are fermionic Gaussian states, and Gaussian transformations preserve this property. Most importantly, a fermionic Gaussian state $\rho$ can be completely characterized by its \textit{covariance matrix}, which is a skew-symmetric real matrix defined by its elements as
        \begin{equation}
            [\Sigma_\rho]_{ij} \coloneqq -\frac{i}{2} \Tr \big[ [m_i, m_j] \rho \big].
        \end{equation}

        We will also use \textit{Gaussian operators}, i.e., operators that are exponentials of some quadratic polynomial in Majorana operators with coefficients that form a complex-valued skew-symmetric matrix. The set comprises both Gaussian states and Gaussian transformations, yet our simulation algorithm involves Gaussian operators that do not fall into either category.
        \begin{proposition}\label{prop:pfaffian_formula}
            Consider a unit-trace Gaussian operator $\hat{O}$, characterized by a covariance matrix $\Sigma$, and a Z-string $Z_{\bm{i}}$ of length $\ell$. Then, we know that
            \begin{equation}
                \Tr [ Z_{\bm{i}} \, \hat{O}] = \pf(\Sigma_{\bm{i}}),
            \end{equation}
            where $\pf$ denotes the Pfaffian defined by
            \begin{equation}
                \pf(A) \coloneqq \frac{1}{2^n n!} \sum_{\sigma \in S_{2n}} \sgn(\sigma) \prod_{i=1}^n A_{\sigma(2i-1), \sigma(2i)}
            \end{equation}
            for a $2n \times 2n$ skew-symmetric matrix $A$.
            Moreover, $\Sigma_{\bm{i}}$ denotes the $2\ell \times 2\ell$ principal submatrix formed by taking the rows and columns of $\Sigma$ according to the indices $2i_1-1, 2i_1, \dots,  2i_\ell-1, 2i_\ell$.
        \end{proposition}
        \begin{proof}
            This fact immediately follows from Wick's theorem~\cite{bravyi2004lagrangianrepresentationfermioniclinear}.
        \end{proof}
        We also know that the Pfaffian can be calculated efficiently on a classical computer:
        \begin{proposition}
            For a $2n \times 2n$ skew-symmetric matrix $A$, its Pfaffian $\operatorname{pf}(A)$ can be computed with a computational cost of $\mathcal{O}(n^3)$.
        \end{proposition}
        \begin{proof}
            The Parlett-Reid algorithm performs a symmetric elimination reducing $A$ to a block-triangular form while preserving its skew-symmetric structure~\cite{Wimmer_2012}. The Pfaffian of a block-triangular matrix can be calculated using $\mathcal{O}(n)$ resources, and the reduction requires $\mathcal{O}(n^3)$ operations, hence, the complexity of  the Pfaffian evaluation is $\mathcal{O}(n^3)$.
        \end{proof}

\section{Efficient classical training strategy}
\label{app:training_strategy}
    In this Appendix, we describe the strategy to train our \fbm\ model efficiently on a classical computer. The strategy relies on a certain Gaussian decomposition of the input quantum state.

\subsection{Input state and ansatz construction}
\label{app:input_state}
    Consider fermionic states produced by $N$ copies of even $4$-mode states, followed by arbitrary fermionic Gaussian transformations, i.e., states of the form
    \begin{equation}\label{eq:generic_ansatz_state}
        \ket{\Psi} = \hat{U} \bigotimes_{i=1}^N \ket{\psi_i},
    \end{equation}
    where $\hat{U}$ is a Gaussian transformation and $\ket{\psi_i}$ are even-parity $4$-mode fermionic quantum states.
    Given a fixed Z-string, its expectation value in such states can be efficiently obtained on a classical computer~\cite{Dias_2024, nest2009simulating}. If we aim to turn such states into an ansatz, we can simplify it using the following fact:\begin{lemma}\label{lemma:gaussian_equivalence}
        Consider an arbitrary $4$-mode even-parity fermionic pure state $\ket{\psi}$. Then there exists a Gaussian unitary $\hat{U}$ such that
        \begin{equation}
            \ket{\psi} = \hat{U} \ket{\alpha},
        \end{equation}
        where        \begin{equation}\label{eq:ket_alpha_definition}
            \ket{\alpha} \coloneqq \cos \alpha \ket{0000} + \sin \alpha \ket{1111} \eqqcolon \cos\alpha \ket{\bar{0}} + \sin\alpha\ket{\bar{1}}.
        \end{equation}
    \end{lemma}
    \begin{proof}
        Since $\ket{\psi}$ is even-parity, it can be written as
        \begin{equation}
            \ket{\psi} = c_0 \ket{\bar{0}} + \sum_{1 \leq i < j \leq 4} c_{ij} f_{i}^\dagger f_{j}^\dagger \ket{\bar{0}} + c_4 f_1^\dagger f_2^\dagger f_3^\dagger f_4^\dagger \ket{\bar{0}}, \quad (c_0, c_{i, j}, c_4 \in \mathbb{C}).
        \end{equation}
        Let $C = ( c_{ij} )_{i, j=1}^4$ be a $4$-by-$4$ matrix containing the $2$-particle amplitudes. A passive Gaussian transformation is a Gaussian transformation which evolves creation operators into a linear combination of creation operators, preserving the total particle number. In detail, for a passive Gaussian transformation $\hat{U}$ this can be written as
        \begin{equation}
            \hat{U} f_{i}^\dagger \hat{U}^\dagger 
            =
            \sum_{k=1}^4 U_{i k} f_k^\dagger,
        \end{equation}
        where $U \in \U(4)$ parametrizes a passive Gaussian transformation.
        Such a transformation acts on a $2$-particle component of the state vector as
        \begin{equation}
            \hat{U} f_{i}^\dagger f_{j}^\dagger \ket{\bar{0}} = \hat{U} f_{i}^\dagger \hat{U}^\dagger \hat{U} f_{j}^\dagger \hat{U}^\dagger \underbrace{\hat{U} \ket{\bar{0}}}_{=\ket{\bar{0}}}
            = \sum_{k,l=1}^4 U_{i k} U_{jl} f_k^\dagger f_l^\dagger \ket{\bar{0}}.
        \end{equation}
        The other components are not mixed, since this transformation preserves the total particle number.
        Then, we can choose $U$ such that it diagonalizes the $2$-particle amplitude matrix, i.e.,
        \begin{equation}
            U^T C U = D
        \end{equation}
        by the Youla decomposition~\cite{Youla_1961}, where $D$ is some diagonal matrix. Hence, there exists a passive linear transformation $\hat{U}$ that
        \begin{equation}
            \hat{U}\ket{\psi} = c'_0 \ket{\bar{0}} + c'_{12} f_{1}^\dagger f_{2}^\dagger \ket{\bar{0}} + c'_{34} f_{3}^\dagger f_{4}^\dagger \ket{\bar{0}} + c'_4 f_1^\dagger f_2^\dagger f_3^\dagger f_4^\dagger \ket{\bar{0}},
        \end{equation}
        with some $c_0', c_{12}', c_{34}', c_{4}' \in \mathbb{C}$.
        To eliminate more terms, consider the gate
        \begin{equation}
            S_{ij}(z) = \exp(z f_i^\dagger f_j^\dagger - \conj{z} f_j f_i).
        \end{equation}
        Since $S_{ij}(z)$ acts on the subspace spanned by the vectors $ \ket{\bar{0}}, f_{i}^\dagger f_{j}^\dagger \ket{\bar{0}} $ with the matrix
        \begin{equation}
            \begin{bmatrix}
                \cos r & -e^{-i\phi} \sin r \\
                e^{i\phi} \sin r & \cos r
            \end{bmatrix},
        \end{equation}
        where $z = r e^{i \phi}$. With such transformation, we can easily eliminate the coefficients $c'_{01}$ and $c'_{23}$. Thus, using an aptly chosen unitary $\hat{V} = S_{12}(z_1)  S_{34}(z_2)$, we get
        \begin{equation}
            \hat{V}\hat{U}\ket{\psi} = c''_0 \ket{\bar{0}} + c_4'' \ket{\bar{1}} \qquad (c_0'', c_4'' \in \mathbb{C})
        \end{equation}
        Finally, the $c''_0$ and $c_4''$ can be rotated to be real and non-negative by a passive linear transformation, yielding the form present in Eq.~\eqref{eq:ket_alpha_definition}.
    \end{proof}
    This means, that every even-parity $4$-mode fermionic state can be parametrized by a single $\alpha \in [0, 2\pi)$ parameter, up to a fermionic Gaussian transformation.
    Using \cref{lemma:gaussian_equivalence}, a state from Eq.~\eqref{eq:generic_ansatz_state} can be written as
    \begin{equation}\label{eq:ansatz}
        \ket{\Psi} = \ket{\bm{\alpha}, \bm{\theta}} \coloneqq U (\bm{\theta}) \ket{\bm{\alpha}},
    \end{equation}
    where $U (\bm{\theta})$ is a Gaussian transformation parametrized by Givens angles as described in Appendix C of Ref.~\cite{oszmaniec2022fermion}, and
    \begin{equation}
        \ket{\bm{\alpha}} \coloneqq \bigotimes_{i=1}^N \ket{\alpha_i},
    \end{equation}
    where $\ket{\alpha_i}$ is defined in Eq.~\eqref{eq:ket_alpha_definition}. This equivalence motivates the adoption of \cref{eq:ansatz} as the input state and ansatz in our construction.

    \subsection{Gaussian operator decomposition}
        To obtain an efficient simulation strategy, we aim to decompose  the state $\ket{\bm{\alpha}, \bm{\theta}}$ into a sum of Gaussian operators. For this, we use the following decomposition of $\ketbra{\alpha}$:
        \begin{proposition}\label{prop:gaussian_decomposition}
            We can decompose $\ketbra{\alpha}$
            into 4 Gaussian operators with nonzero trace.
        \end{proposition}
        \begin{proof}
            We can write
             \begin{equation}
                \ketbra{\alpha} = \cos^2(\alpha) \ketbra{\bar{0}} + \cos (\alpha) \sin (\alpha) \left[\ketbra{\bar{0}}{\bar{1}}+\ketbra{\bar{1}}{\bar{0}}\right] + \sin^2(\alpha) \ketbra{\bar{1}}.
            \end{equation}
            Let us define
            \begin{equation}
                \ket{\chi} \coloneqq \ket{1100},
            \end{equation}
            and
            \begin{subequations}
            \begin{align}
                 \ket{\Phi_0} &\coloneqq \cos (\alpha) \ket{\bar{0}} + \ket{\chi},\\
                 \ket{\Phi_1} &\coloneqq \sin (\alpha) \ket{\bar{1}} - \ket{\chi}.
            \end{align}
            \end{subequations}
            We know that $\ket{\alpha} = \ket{\Phi_0} + \ket{\Phi_1}$, and hence
            \begin{equation}
                \ketbra{\alpha}{\alpha} = \ketbra{\Phi_0}{\Phi_0} + \ketbra{\Phi_0}{\Phi_1} + \ketbra{\Phi_1}{\Phi_0} + \ketbra{\Phi_1}{\Phi_1}.
            \end{equation}
            It can be easily verified that $\bra{\Phi_j}\ket{\Phi_k} \neq 0$.
            Introducing $y = \arccot( \cos(\alpha))$, we can show that $\ketbra{\Phi_j}{\Phi_k}$ are Gaussian operators. If $\alpha = 0$, it is trivial for $\ket{\Phi_0}$, otherwise, when $\alpha \neq 0$, we can show that
            \begin{align}
                \ket{\Phi_0} = \frac{1}{\sin y} \exp((\pi/2-y) m_1 m_3) \ket{\chi} = \cot y \ket{\bar{0}} + \ket{\chi} = \cos \alpha \ket{\bar{0}} + \ket{\chi}
            \end{align}
            Similarly, we can take $y' = \arctan \sin \alpha$
            \begin{equation}
                \ket{\Phi_1} = \frac{1}{\cos y'} \exp((\pi-y') m_5 m_7) \ket{\chi} = \tan y' \ket{\bar{1}} - \ket{\chi} = \sin \alpha \ket{\bar{1}} - \ket{\chi}.
            \end{equation}
            Since $\ketbra{\chi}$ is a Gaussian state, it follows that $\ketbra{\Phi_i}{\Phi_j}$ is a Gaussian operator, since the product of Gaussian operators is Gaussian itself.
        \end{proof}
    
        To proceed, we set out to obtain the covariance matrices of these $4$ Gaussian operators.
        We can calculate the normalizations as
        \begin{subequations}
        \begin{align}
            \braket{\Phi_0}{\Phi_0} &= \cos^2 \alpha + 1, \\
            \braket{\Phi_1}{\Phi_1} &= \sin^2 \alpha + 1, \\
            \braket{\Phi_0}{\Phi_1} &= \braket{\Phi_1}{\Phi_0} = -1.
        \end{align}
        \end{subequations}
        Using these relations, we can normalize the Gaussian constituents by defining
        \begin{equation}
            \rho_{ij}(\alpha) \coloneqq \frac{\ketbra{\Phi_i}{\Phi_j}}{\braket{\Phi_i}{\Phi_j}},
        \end{equation}
        yielding the $\rho_{ij}(\alpha)$ unit-trace Gaussian operators.
        Now we aim to calculate the second-order correlators of these Gaussian operators.
        When $j \neq k$, we know that
        \begin{equation}
            \bra{\bar{0}} m_j m_k \ket{\bar{1}} = 0 \qquad (j \neq k),
        \end{equation}
        and that
        \begin{align}
        \begin{split}
            \bra{\bar{0}} m_{2j-1} m_{2j} \ket{\bar{0}} &= i,\\
            \bra{\bar{1}} m_{2j-1} m_{2j} \ket{\bar{1}} &= -i,
        \end{split}
        \end{align}
        but $0$ in other cases.
        We also know that 
        \begin{align}
        \begin{split}
            \bra{\bar{0}} m_1 m_3 \ket{\chi} &= -1, \\
            \bra{\bar{0}} m_1 m_4 \ket{\chi} &= i, \\
            \bra{\bar{0}} m_2 m_3 \ket{\chi} &= i, \\
            \bra{\bar{0}} m_2 m_4 \ket{\chi} &= 1,
        \end{split}
        \end{align}
        and that
        \begin{align}
        \begin{split}
            \bra{\bar{1}} m_5 m_7 \ket{\chi} &= -1, \\
            \bra{\bar{1}} m_5 m_8 \ket{\chi} &= -i, \\
            \bra{\bar{1}} m_6 m_7 \ket{\chi} &= -i, \\
            \bra{\bar{1}} m_6 m_8 \ket{\chi} &= 1,
        \end{split}
        \end{align}
        otherwise $0$.
        Moreover, we know that 
        \begin{equation}
            \bra{\chi} m_{2j-1} m_{2j} \ket{\chi} = -i(-1)^{\lfloor j/2 \rfloor} =\begin{cases}
                -i &\qquad j=1, 2,\\
                i &\qquad j = 3, 4,
            \end{cases}
        \end{equation}
        and otherwise every second moment is $0$.
        From this, the covariance matrix of the Gaussian operators $\rho_{00}, \rho_{11}$ and $\rho_{01}$ can be easily calculated. 
        Given a fermionic quantum state modeled by a density matrix $\rho$, the matrix elements of the covariance matrix are defined as
        \begin{equation}
            [\Sigma_\rho]_{ij} \coloneqq -\frac{i}{2} \Tr \big[ [m_i, m_j] \rho \big] = -i \Tr \big[ m_i m_j \rho \big] \qquad (i \neq j).
        \end{equation}
        Putting it all together, the covariance matrices for $\rho_{00}(\alpha), \rho_{11}(\alpha)$ and $\rho_{01}(\alpha)$ are
        \begin{subequations}
        \begin{gather}
            \Sigma_{00}(\alpha) = \frac{1}{1+c^2} \begin{bmatrix}
                  & -s^2 &   & 2c &   &   &   &  \\
                s^2 &   & 2c &   &   &   &   &  \\
                  & -2c &   & -s^2 &   &   &   &  \\
                -2c &   & s^2 &   &   &   &   &  \\
                  &   &   &   &   & 1+c^2 &   &  \\
                  &   &   &   & -(1+c^2) &   &   &  \\
                  &   &   &   &   &   &   & 1+c^2 \\
                  &   &   &   &   &   & -(1+c^2) &  \\
            \end{bmatrix},
            \\
            \Sigma_{11}(\alpha) = \frac{1}{1+s^2} \begin{bmatrix}
                  & -(1+s^2) &   &   &   &   &   &  \\
                1+s^2 &   &   &   &   &   &   &  \\
                  &   &   & -(1+s^2) &   &   &   &  \\
                  &   & 1+s^2 &   &   &   &   &  \\
                  &   &   &   &   & c^2 &   & -2s \\
                  &   &   &   & -c^2 &   & -2s &  \\
                  &   &   &   &   & 2s &   & c^2\\
                  &   &   &   & 2s &   & -c^2 &  \\
            \end{bmatrix},
            \\
            \Sigma_{01}(\alpha) = \begin{bmatrix}
                  & -1 & -i c & c &   &   &   &  \\
                1 &   & c & i c &   &   &   &  \\
                i c & -c &   & -1 &   &   &   &  \\
                -c & -i c & 1 &   &   &   &   &  \\
                  &   &   &   &   & 1 & i s & - s\\
                  &   &   &   & -1 &   & -s & -i s \\
                  &   &   &   & -i s & s &   & 1 \\
                  &   &   &   & s & i s & -1 &  \\
            \end{bmatrix},
        \end{gather}
        \end{subequations}
        where $c \coloneqq \cos(\alpha)$ and $s \coloneqq \sin(\alpha)$.

    \subsection{Calculating Z-string expectation values}
    \label{app:Z_expvals}
        Our estimation algorithm relies on factoring the input state into Gaussian operators as presented in the previous section.

        We start by analyzing $\ketbra{\alpha}$. While the state $\ketbra{\alpha}$ itself is not Gaussian, according to \cref{prop:gaussian_decomposition}, we can decompose it as 
        \begin{equation}
            \ketbra{\alpha} = \rho_{\text{Gauss}}(\alpha) + \sigma(\alpha),
        \end{equation}
        where $\rho_{\text{Gauss}}$ denotes the Gaussian state corresponding to the covariance matrix of $\ketbra{\alpha}$ (which we will denote by $\Sigma_{\text{Gauss}}(\alpha)$), and
        \begin{equation}
            \sigma(\alpha) \coloneqq - \rho_{\text{Gauss}}(\alpha) + (\cos^2 \alpha + 1) \rho_{00}(\alpha) + (\sin^2 \alpha + 1) \rho_{11}(\alpha) 
            - \rho_{01}(\alpha) - \rho_{10}(\alpha).
        \end{equation}
        We know that for a single $Z_i$ operator, the expectation value is
        \begin{equation}
            \Tr [ \ketbra{\alpha} Z_i ] =  \Tr [ \rho_{\text{Gauss}}(\alpha) Z_i ],
        \end{equation}
        as the term $\sigma(\alpha)$ contributes nothing, i.e., $\Tr [\sigma(\alpha) Z_i] = 0$ for all $Z$-operators. For $Z$-strings of length $2$ or more, both $\rho_{\text{Gauss}}(\alpha)$ and $\sigma(\alpha)$ contribute. Considering such an $\ell$-long $Z$-string denoted by $Z_{\bm{i}} \coloneqq Z_{i_0} \dots Z_{i_{\ell}}$, we can directly write
        \begin{align}
            \Tr[ \ketbra{\alpha} Z_{\bm{i}}] =
            \sum_{k, l = 0}^1
            N_{k l}(\alpha) \Tr[ \rho_{k l}(\alpha) Z_{\bm{i}}].
        \end{align}
        where
        \begin{equation}
        \begin{split}
            N_{00}(\alpha) &= \cos^2 \alpha + 1, \\
            N_{11}(\alpha) &= \sin^2 \alpha + 1, \\
            N_{01}(\alpha) &= N_{10}(\alpha) = -1.
        \end{split}
        \end{equation}
        According to \cref{prop:gaussian_decomposition}, $\rho_{kl}$ are Gaussian operators with nonzero trace, and hence, we can calculate $\Tr[ \rho_{kl}(\alpha) Z_{\bm{i}}]$ efficiently using \cref{prop:pfaffian_formula} for a fixed-length Z-string $Z_{\bm{i}}$ using their covariance matrix.
        \begin{lemma}
        \label{lemma:simulation_complexity}
            Given a fermionic Gaussian transformation $U(\bm{\theta})$, a Pauli-Z string $Z_{\bm{i}}$ with fixed length $\ell \coloneqq | \bm{i}|$, and an input state $\ket{\bm{\alpha}} = \bigotimes_{j=1}^N (\cos \alpha_j \ket{0000} + \sin \alpha_j \ket{1111} )$, expectation values of the form 
            \begin{equation}
                \langle Z_{\bm{i}} \rangle_{\bm{\alpha}, \bm{\theta}} = \Tr[ Z_{\bm{i}} \ketbra{\bm{\alpha}, \bm{\theta}}]
            \end{equation}
            can be computed in time $\mathcal{O}(N^{\lfloor \frac{\ell}{2} \rfloor})$, where $N$ is the number of subsystems and $\ket{\bm{\alpha}, \bm{\theta}} = U(\bm{\theta})\ket{\bm{\alpha}}$.
        \end{lemma}
        \begin{proof}
            First, observe that the length-$1$ Z-string expectation values can be directly calculated from $
            \bigotimes_{i=1}^{N}\rho_{\text{Gauss}}(\alpha_i)$, since the remaining terms in $\rho(\alpha_i)$ have zero contribution overall and the Gaussian time evolution preserves the length of Majorana strings. More concretely, for length-$0$ and length-$1$ Z-strings, we can write
            \begin{equation}
                 \langle Z_{\bm{i}} \rangle_{\bm{\alpha}, \bm{\theta}}
                 = 
                 C^{(0)}_{\bm{i}}(\bm{\alpha}, \bm{\theta}) \qquad (|\bm{i}| = 0, 1),
            \end{equation}
            where we define
            \begin{equation}
                C^{(0)}_{\bm{i}}(\bm{\alpha}, \bm{\theta}) =   \Tr \left[
                    U(\bm{\theta})
                    \left(
                    \bigotimes_{k=1}^{N}\rho_{\text{Gauss}}^{(k)}(\alpha_k)
                    \right)
                    U^\dagger(\bm{\theta})
                    Z_{\bm{i}}
                \right].
            \end{equation}
            Moreover, for length-$2$ and length-$3$ Z-string expectation values, we need to include corrections from terms containing a single $\rho_{ab}(\alpha)$ density matrices as well, i.e.,
            \begin{equation}
                \langle Z_{\bm{i}} \rangle_{\bm{\alpha}, \bm{\theta}} =  
                C^{(0)}_{\bm{i}}(\bm{\alpha}, \bm{\theta})
                +
                C^{(1)}_{\bm{i}}(\bm{\alpha}, \bm{\theta}) \qquad (|\bm{i}| = 2, 3),
            \end{equation}
            where $C^{(1)}_{\bm{i}}(\bm{\alpha}, \bm{\theta})$ is the first-order correction defined by
            \begin{equation}
                C^{(1)}_{\bm{i}}(\bm{\alpha}, \bm{\theta}) \coloneqq \sum_{j = 1}^{N}
                \sum_{a, b = 0}^1
                N_{a b}(\alpha_k)
                \Tr \left[
                    U(\bm{\theta})
                    \left(
                    \bigotimes_{k\neq j}\rho_{\text{Gauss}}(\alpha_k)
                    \otimes \rho_{a b}(\alpha_j)\right)
                    U^\dagger(\bm{\theta})
                    Z_{\bm{i}}
                \right].
            \end{equation}
            More generally, the $L$-order correction is defined by
            \begin{equation}
                C^{(L)}_{\bm{i}}(\bm{\alpha}, \bm{\theta}) 
                \coloneqq
                \sum_{\substack{\bm{j} \in \mathcal{P}([N])\\ | \bm{j}| = L}}
                \sum_{\bm{a}, \bm{b} \in \{ 0, 1\}^L}
                \left(
                \prod_{s=1}^L N_{a_s b_s}(\alpha_{s})
                \right)
                E_{\bm{i}, \bm{j};\bm{a},\bm{b}} (\bm{\alpha}, \bm{\theta} ),
            \end{equation}
            where
            \begin{align}
                E_{\bm{i}, \bm{j};\bm{a},\bm{b}} (\bm{\alpha}, \bm{\theta} )
                &\coloneqq
                \Tr \left[
                    U(\bm{\theta})
                        \rho_{\bm{j};\bm{a},\bm{b}}(\bm{\alpha})
                    U^\dagger(\bm{\theta})
                    Z_{\bm{i}}
                \right], \\
                \rho_{\bm{j};\bm{a},\bm{b}}(\bm{\alpha})
                &=
                \bigotimes_{k \notin \bm{j}}
                    \rho_{\text{Gauss}}^{(k)}(\alpha_k)
                    \otimes
                    \bigotimes_{s=1}^L
                \rho_{a_s b_s}^{(s)}(\alpha_{s}).
            \end{align}
            $C^{(L)}_{\bm{i}}(\bm{\alpha}, \bm{\theta})$ is the required correction term for length-$(2L)$ and length-$(2L+1)$ Z-strings, i.e.,
            \begin{equation}
                \langle Z_{\bm{i}} \rangle_{\bm{\alpha}, \bm{\theta}} =  
                \sum_{l=0}^L C^{(l)}_{\bm{i}}(\bm{\alpha}, \bm{\theta}) \qquad (|\bm{i}| = 2L, 2L+1).
            \end{equation}
            Each sum in this expression iterates over polynomial number of terms in $N$ with fixed Z-string-length $\ell$ (amounting to fixed $L$),
            and the quantity can be calculated using \cref{prop:pfaffian_formula} by noticing that
            $
                 U(\bm{\theta})
                     \rho_{\bm{j};\bm{a},\bm{b}}(\bm{\alpha})
                    U^\dagger(\bm{\theta})
            $
            is a Gaussian operator with covariance matrix
            \begin{align}\label{eq:covmat_compose}
                O \left(
                    \bigoplus_{i \notin \bm{j}} \Sigma_{\text{Gauss}}^{(k)}(\alpha_k) \oplus 
                    \bigoplus_{s=1}^L
                    \Sigma_{a_s b_s}^{(s)}(\alpha_{s})
                \right) O^T
               & =
                \sum_{i \notin \bm{j}} O (\Sigma_{\text{Gauss}}^{(k)}(\alpha_k) \oplus 0_{2d-8} ) O^T \\
                & + 
                \sum_{s=1}^L
                O (
                    \Sigma_{a_s b_s}^{(s)}(\alpha_{s}) \oplus 0_{2d-8}
                ) O^T, \nonumber
            \end{align}
            where $O \in \mathrm{SO}(2d)$ is determined by Eq.~\eqref{eq:so_matrix}.
            Notice, that each direct summand inside the parentheses can be calculated in advance with $\mathcal{O}(N^2)$ operations, and for each 
            $E_{\bm{i}, \bm{j};\bm{a},\bm{b}} (\bm{\alpha}, \bm{\theta} )$, we have to add together $N$ such matrices, costing $\mathcal{O}(N^2)$ operations in total.
            
            Using these formulas, the $L$-order correction can be calculated using in $\mathcal{O}(L^3 4^L N^L)$, which is $\mathcal{O}(N^{\lfloor\frac{\ell}{2}\rfloor})$ for Z-strings of fixed length $\ell$.
        \end{proof}

        In the following, we give an algorithm for computing the expectation values for all $Z_{\bm{i}}$ where $\bm{i}$ is at most length $\ell$.

        \begin{algorithm}[H]
        \caption{Calculation of Z-string expectation values}
        \label{alg:zstring_expectation}
        \begin{algorithmic}
        
        \Require Parameters $\bm{\alpha} = (\alpha_1, \dots, \alpha_N)$, Gaussian transformation $U(\bm{\theta})$, and maximal Z-string length $\ell$, number of registers $N$,

        \vspace{0.5em}
        \State \textbf{Precomputation:}
        \For{$k = 1$ to $N$}
            \State Compute and store $\Sigma_{\text{Gauss}}(\alpha_k)$ and $\Sigma_{a b}(\alpha_k)$ for $a, b = 0, 1$.
            \State Compute coefficients:
            \[
                N_{00}(\alpha_k) \gets \cos^2 \alpha_k + 1,\quad
                N_{11}(\alpha_k) \gets \sin^2 \alpha_k + 1,\quad
                N_{01}(\alpha_k) \gets N_{10}(\alpha_k) \gets -1.
            \]
        \EndFor
        
        \vspace{0.5em}
        \State \textbf{Correction terms:}
        \State $T \gets$ array indexed by Z-string indices $\bm{i}$, initialized with zeros.
        \For{$L = 0$ to $\lfloor \ell/2 \rfloor$}
            \ForAll{$\bm{j} \in \mathcal{P}([N])$ with $|\bm{j}| = L$}
                \ForAll{$(\bm{a}, \bm{b}) \in \{0,1\}^L$}
                    \State $N \gets \prod_{s=1}^{L} N_{a_s b_s}(\alpha_s).$
                    \State $\Sigma \gets$ covariance matrix of
                    \(
                        U(\bm{\theta})
                        \rho_{\bm{j};\bm{a},\bm{b}}(\bm{\alpha})
                        U^\dagger(\bm{\theta})
                    \) from \cref{eq:covmat_compose}.
                    \ForAll{$\bm{i}$ with $|\bm{i}| \le 2L+1$}
                        \State $T_{\bm{i}} \gets T_{\bm{i}} +  N \pf(\Sigma_{\bm{i}})$.
                    \EndFor
                \EndFor
            \EndFor
        \EndFor
        \State \Return Expectation values $\langle Z_{\bm{i}} \rangle_{\bm{\alpha}, \bm{\theta}} \coloneqq T_{\bm{i}}$ for all $\bm{i}$ with $|\bm{i}| \leq \ell$.
        \end{algorithmic}
        \end{algorithm}

        \subsection{Estimating the \mmd\ loss}
        \label{app:estimating_mmd}
        We have established the cost of estimating expectation values for the model probability distribution and have given an explicit algorithm. However, to estimate the \mmd\ loss function, we also need estimates of the target distribution based on training samples. Luckily, this can be computed efficiently by computing the mean parity on the given bits. As shown in \cite{recio2025train}, this can be formulated as
        \begin{equation}
            \expval{Z_{\bm{i}}}_p \approx \frac{1}{|\mathcal{X}|}\sum_{\bm{x} \in \mathcal{X}} (-1)^{x_{\bm{i}}}.
        \label{eq:target_expval}
        \end{equation}
        Now we can formulate an algorithm for estimating the loss function. 
        \begin{algorithm}
            \caption{\mmd\ computation from sampled Z-strings}
            \label{alg:mmd_computation}
            \begin{algorithmic} 
                
                \Require 
                    List of distributions $\mathcal{P} = \{p_{K_1}, p_{K_2}, \dots\}$;
                    Number of operators per distribution $n_{\text{ops}}$;
                    Target data $\mathcal{D}$;
                    Model parameters $\bm{\alpha}, \bm{\theta}$.
                
                \State $\mathcal{Z} \gets [\,]$   
                \For{each $p_{K_i} \in \mathcal{P}$}
                    \State $S_i \gets$ $n_{\text{ops}}$ Z-strings samples from $p_{K_i}$.
                    \State $\mathcal{Z} \gets \mathcal{Z} \cup S_i$ 
                \EndFor
                
                \State $E_{\text{target}} \gets [\,]$
                
                \For{each $Z \in \mathcal{Z}$}
                    \State $e_t \gets$ expectation value of $Z$ computed from $\mathcal{D}$ using \cref{eq:target_expval}.
                    \State Append $e_t$ to $E_{\text{target}}$
                \EndFor
                    
                \State $E_{\text{model}} \gets$ model expectation value of all $Z \in \mathcal{Z}$ based on $(\bm{\alpha},\bm{\theta})$ using \cref{alg:zstring_expectation}.
        
                \State $L_{\text{MMD}^2} \gets \textproc{MMD}^2(E_{\text{target}}, E_{\text{model}})$
                using \cref{eq:mmd_expval}
                \State \Return $L_{\text{MMD}^2}$.      
            \end{algorithmic}
        \end{algorithm}

\section{Absence of barren plateaus} \label{app:barren_plateau}
    As demonstrated in Refs.~\cite{cerezo2025does, diaz2023showcasing}, when the input state is fixed and the measured Pauli observable has constant locality, the Gaussian transformation $\uflo(\bm{\theta})$ does not lead to barren plateaus in the parameters $\bm{\theta}$. In this discussion, we parametrize the FLO circuit with $O \in \text{SO}(2d)$ instead of the Givens angles $\bm{\theta}$, and denote the unitary of the reparametrized circuit as $\uflo(\bm{\theta}) = U(O)$.

    Let us denote the expectation value of $Z_{\bm{i}}$ by
    \begin{equation}
        f_{\bm{i}}(\bm{\alpha}, O) 
        \coloneqq 
        \Tr \left[
            \ketbra{\bm{\alpha}} U(O)^\dagger Z_{\bm{i}} U(O)
        \right].
    \end{equation}
    To determine whether our ansatz has a barren plateau in the parameters $\bm{\alpha}$, we need to investigate the variance of the gradient of the expectation values of $Z_{\bm{i}}$. More concretely, we are interested in the quantity
    \begin{equation}
        \underset{
            \substack{
            \alpha_i \sim [0, 2\pi)
            \\
            O \sim \text{Haar(SO($2d$))}
            }
        }{\text{Var}} [ \partial_{\alpha_j} f_{\bm{i}}(\bm{\alpha}, O)]
        =
        \underset{
            \substack{
            \alpha_i \sim [0, 2\pi)
            \\
            O \sim \text{Haar(SO($2d$))}
            }
        }{\mathbb{E}} [ 
        \partial_{\alpha_j}f_{\bm{i}}(\bm{\alpha}, O)^2]
        - 
            \underset{
            \substack{
            \alpha_i \sim [0, 2\pi)
            \\
            O \sim \text{Haar(SO($2d$))}
            }
        }{\mathbb{E}} [ 
        \partial_{\alpha_j}f_{\bm{i}}(\bm{\alpha}, O)]
        ^2.
    \end{equation}
    The expectation value $f_i$ shows a barren plateau when $\text{Var}[ \partial_{\alpha_j}f_{\bm{i}}(\bm{\alpha}, O)] \in \mathcal{O}(b^{-n})$ for some $b > 1$.
    We can remove the partial derivatives simply using the fact that $\partial_{\alpha_j} f_{\bm{i}}(\bm{\alpha}, O) = f_{\bm{i}}(\bm{\alpha} + (\pi/2) \bm{e}_j, O)$, where $\bm{e}_j$ is the $j$-th canonical basis vector, and use the shift-invariance of the average. Moreover, we know that the first moment of $f_{\bm{i}}$ is $0$. Therefore, we only need to calculate the second moment
    \begin{equation}
        \underset{
            \substack{
            \alpha_i \sim [0, 2\pi)
            \\
            O \sim \text{Haar(SO($2d$))}
            }
        }{\mathbb{E}} [ 
        f_{\bm{i}}(\bm{\alpha}, O)^2]
        =
        \Tr \left[
            \underset{
                \alpha_i \sim [0, 2\pi)
            }{\mathbb{E}} [\ketbra{\bm{\alpha}}^{\otimes 2}]
            \underset{
            O \sim \text{Haar(SO($2d$))}
            }{\mathbb{E}}
            [(U(O)^\dagger Z_{\bm{i}} U(O))^{\otimes 2}]
        \right]
    \end{equation}
    We can write
    \begin{equation}
        \underset{\alpha \sim [0, 2\pi)}{\mathbb{E}}\left[ \ketbra{\alpha}^{\otimes 2} \right]
        = \frac{1}{2} (\ketbra{0000}^{\otimes 2} + \ketbra{1111}^{\otimes 2})
    \end{equation}
    and denote
    \begin{equation}
        M_{\bm{i}} \coloneqq \underset{
            O \sim \text{Haar(SO($2d$))}
            }{\mathbb{E}}
            [(U(O)^\dagger Z_{\bm{i}} U(O))^{\otimes 2}]
    \end{equation}
    for brevity. Then
    \begin{equation}
        \underset{
            \substack{
            \alpha_i \sim [0, 2\pi)
            \\
            O \sim \text{Haar(SO($2d$))}
            }
        }{\mathbb{E}} [ 
        f_{\bm{i}}(\bm{\alpha}, O)^2]
        =
        \frac{1}{2^N}
        \sum_{\bm{x} \in \{0, 1\}^N}
        \Tr \left[
            \ketbra{\bar{\bm{x}}}^{\otimes 2}
            M_{\bm{i}}
        \right],
    \end{equation}
    where $\ket{\bar{\bm{x}}} \coloneqq \bigotimes_{i=1}^N \ket{x_i x_i x_i x_i}$.
    Now, for each $\bm{x}$, we can find an FLO unitary $U_{\bm{x}}$ that prepares $\ket{\bar{\bm{x}}}$ from the vacuum, i.e., $\ket{\bar{\bm{x}}} = U_{\bm{x}} \ket{\bm{0}}$. Finally, we can use the fact that $M_{\bm{i}}$ commutes with every FLO transformation (specifically, with every $U_{\bm{x}}$), and get that the variance reduces to the case already familiar from Ref.~\cite{diaz2023showcasing}:
    \begin{equation}
         \underset{
            \substack{
            \alpha_i \sim [0, 2\pi)
            \\
            O \sim \text{Haar(SO($2d$))}
            }
        }{\text{Var}} [ \partial_{\alpha_j} f_{\bm{i}}(\bm{\alpha}, O)]
        = 
        \underset{
            O \sim \text{Haar(SO($2d$))}
        }{\text{Var}} [ \Tr [\ketbra{\bm{0}} U(O)^\dagger Z_{\bm{i}} U(O) ]] \notin \mathcal{O}(b^{-n}).
    \end{equation}
    This shows that introducing non-Gaussian resources in the form of parametrized magic input states does not make the loss landscape harder to navigate.
    Moreover, assuming that the differences $\expval{Z_{\bm{i}}}_p -\expval{Z_{\bm{i}}}_{q_{\mathbf{w}}}$ at initialization scale as $\mathcal{O}(1/\text{poly}(N))$, barren plateaus are absent in the loss function \mmd\ as well.

    We can also use a different argument, which is more suited for our setup. If the circuit parameters $\boldsymbol{\theta}$ are initialized uniformly at random over $[0,2\pi)$, rather than drawn from the Haar measure, it is better to view the loss function as a trigonometric polynomial in the angles $\alpha_i$ and $\theta_j$. The number of such variables is $\mathcal{O}(\mathrm{poly}(N))$, and, because each $Z$-string has $\mathcal{O}(1)$ length, the total degree of the polynomial and the number of monomials in the polynomial are also $\mathcal{O}(\mathrm{poly}(N))$. Consequently, the variance of the loss function (as well as that of its partial derivatives with respect to $\alpha_i$ and $\theta_j$) reduces to computing expectations of products of powers of sine and cosine functions. Using the fact that
    $
         \mathbb{E}_{\theta\sim[0,2\pi)}[\cos^n(\theta)] 
        = \mathcal{O}(n^{-1/2}),
    $
    one concludes that each monomial contributes at most inverse-polynomial weight. Since the loss contains only polynomially many terms, its variance---and the variance of its first derivatives---scales as $\mathcal{O}(1/\mathrm{poly}(N))$.

\end{document}